\documentclass[11pt]{article}
  
\usepackage[letterpaper,margin=1in]{geometry}
\usepackage{amsmath}
\usepackage{amsthm}
\usepackage{amssymb}
\usepackage[inline]{enumitem}
\usepackage{mathtools}
\usepackage{bm}
\usepackage{microtype}
\usepackage{xfrac}

\usepackage{times}

\usepackage{url}
\usepackage{hyperref}
\hypersetup{colorlinks, linkcolor=darkblue, citecolor=darkgreen, urlcolor=darkblue}

\usepackage[framemethod=tikz]{mdframed}

\usepackage[vlined,ruled,algo2e]{algorithm2e}
\SetAlCapHSkip{0.2em}
\SetCustomAlgoRuledWidth{0.95\textwidth}
\makeatletter
\patchcmd{\@algocf@start}{%
  \begin{lrbox}{\algocf@algobox}%
}{%
  \rule{0.025\textwidth}{\z@}%
  \begin{lrbox}{\algocf@algobox}%
  \begin{minipage}{0.95\textwidth}%
}{}{}
\patchcmd{\@algocf@finish}{%
  \end{lrbox}%
}{%
  \end{minipage}%
  \end{lrbox}%
}{}{}
\makeatother

\definecolor{darkblue}{rgb}{0,0,0.38}
\definecolor{darkgreen}{rgb}{0.1,0.35,0}

\DeclareMathOperator{\comp}{comp}

\makeatletter
\newcommand{\labeltarget}[1]{\Hy@raisedlink{\hypertarget{#1}{}}}
\makeatother

\makeatletter
\let\@@pmod\pmod
\DeclareRobustCommand{\pmod}{\@ifstar\@pmods\@@pmod}
\def\@pmods#1{\mkern8mu({\operator@font mod}\mkern 6mu#1)}
\makeatother

\makeatletter
\let\@@mod\mod
\DeclareRobustCommand{\mod}{\@ifstar\@mods\@@mod}
\def\@mods#1{\mkern8mu{\operator@font mod}\mkern 6mu#1}
\makeatother

\def\argmax{\operatorname{argmax}}
\def\argmin{\operatorname{argmin}}

\newtheorem{theorem}{Theorem}[section]
\newtheorem{lemma}[theorem]{Lemma}
\newtheorem{definition}[theorem]{Definition}
\newtheorem{corollary}[theorem]{Corollary}

\title{Submodular Minimization Under Congruency Constraints}

\author{
Martin N\"agele\thanks{
Department of Mathematics, ETH Zurich, Zurich, Switzerland.
Email: \href{mailto:martin.naegele@ifor.math.ethz.ch}%
{martin.naegele@ifor.math.ethz.ch}.}
\and 
Benny Sudakov \thanks{
Department of Mathematics, ETH Zurich, Zurich, Switzerland.
Email: \href{mailto:benjamin.sudakov@math.ethz.ch}%
{benjamin.sudakov@math.ethz.ch}.
Research supported in part by the Swiss National Science Foundation grant 200021\_175573.%
}
\and
Rico Zenklusen\thanks{
Department of Mathematics, ETH Zurich, Zurich, Switzerland.
Email: \href{mailto:ricoz@math.ethz.ch}%
{ricoz@math.ethz.ch}.
Research supported in part by the Swiss National Science Foundation grant 200021\_165866.%
}
}
 
\begin{document}

\maketitle

\begin{abstract}
Submodular function minimization (SFM) is a fundamental and efficiently solvable problem in combinatorial optimization with a multitude of applications in various fields.
Surprisingly, there is only very little known about constraint types under which SFM remains efficiently solvable. The arguably most relevant non-trivial constraint class for which polynomial SFM algorithms are known are parity constraints, i.e., optimizing only over sets of odd (or even) cardinality. Parity constraints capture classical combinatorial optimization problems like the odd-cut problem, and they are a key tool in a recent technique to efficiently solve integer programs with a constraint matrix whose subdeterminants are bounded by two in absolute value.

We show that efficient SFM is possible even for a significantly larger class than parity constraints, by introducing a new approach that combines techniques from Combinatorial Optimization, Combinatorics, and Number Theory. In particular, we can show that efficient SFM is possible over all sets (of any given lattice) of cardinality $r \bmod{m}$, as long as $m$ is a constant prime power. This covers generalizations of the odd-cut problem with open complexity status, and has interesting links to integer programming with bounded subdeterminants.
To obtain our results, we establish a connection between the correctness of a natural algorithm, and the nonexistence of set systems with specific combinatorial properties. We introduce a general technique to disprove the existence of such set systems, which allows for obtaining extensions of our results beyond the above-mentioned setting.
These extensions settle two open questions raised by Geelen and Kapadia [Combinatorica, 2017] in the context of computing the girth and cogirth of certain types of binary matroids.
\end{abstract}
 
\medskip

\newpage

\section{Introduction}
Submodular function minimization (SFM) is a central combinatorial optimization problem with numerous applications in many fields, including speech analysis, image segmentation, combinatorial optimization, and integer programming (see~\cite{schrijver_2003_combinatorial,mccormick_2005_submodular,iwata_2008_submodular,chakrabarty_2017_subquadratic,artmann_2017_strongly} and references therein).
A set function $f\colon2^N \rightarrow \mathbb{R}$ on a finite ground set $N$ is submodular if
\begin{equation*}
f(A)+f(B)\geq f(A\cup B)+f(A\cap B)
\quad \forall A,B\subseteq N\enspace.
\end{equation*}
The high relevance of SFM is explained by the fact that the above condition, which defines submodularity, is equivalent to the diminishing returns property, which is a very natural property of set functions appearing in various contexts.%
\footnote{%
A set function $f\colon2^N \rightarrow \mathbb{R}$ on a finite ground set $N$ satisfies the diminishing returns property if $f(A\cup\{e\}) - f(A) \geq f(B\cup \{e\}) -f(B)$ for all $A\subseteq B\subseteq N$ and $e\in N\setminus B$.
For more information on submodular functions, we refer the interested reader to~\cite{schrijver_2003_combinatorial,mccormick_2005_submodular,fujishige_2005_submodular}.
}
Typical examples of submodular functions include valuation functions in economics, cut functions, matroid rank functions, the Shannon entropy of joint distributions, and coverage functions, just to name a few.

A cornerstone result in Combinatorial Optimization, known since the early '80s, is that SFM is efficiently solvable, only assuming value oracle access to the submodular function~\cite{groetschel_1981_ellipsoid}, which is the usual model in the field and assumed throughout this paper. Typically, results on SFM easily carry over to lattices, implying that efficient SFM is possible over any lattice of the ground set.\footnote{A lattice $\mathcal{L}\subseteq 2^N$ over a ground set $N$ is a set family that is closed under unions and intersections, i.e., for any $A,B\in \mathcal{L}$, we have $A\cup B, A\cap B \in \mathcal{L}$. Whenever a lattice is given, we make the standard assumption that it is given by a compact encoding in terms of a digraph (see~\cite[Section 10.3]{groetschel_1993_geometric}).
What we call a lattice is sometimes also called a lattice family, a ring family, or a distributive lattice.
}
Since the early results on SFM, there has been exciting progress on the subject with some recent impressive speedups in the best-known running times for solving SFM~\cite{cunningham_1985_submodular,schrijver_2000_combinatorial,iwata_2001_combinatorial,iwata_2009_simple,lee_2015_faster,chakrabarty_2017_subquadratic}.

Unfortunately, the picture is much less satisfactory for constrained SFM. A canonical extension of the unconstrained case is obtained by only considering non-empty sets, a problem that can easily be reduced to unconstrained SFM by guessing one element of an optimal solution. Another relatively direct extension that includes the case of non-empty sets is that SFM is efficiently solvable over intersecting or crossing set families (see~\cite[Volume B]{schrijver_2003_combinatorial}).%
\footnote{A set family $\mathcal{F}\subseteq N$ is \emph{intersecting} if for any $A,B\subseteq \mathcal{F}$ such that $A\setminus B, B\setminus A, A\cap B \neq \emptyset$, we have $A\cup B, A\cap B\in \mathcal{F}$. Moreover, $\mathcal{F}$ is \emph{crossing} if for any $A,B\in \mathcal{F}$ with $A\setminus B, B\setminus A, A\cap B, N\setminus (A\cup B)\neq\emptyset$, we have $A\cup B, A\cap B \in \mathcal{F}$.}
Surprisingly, very little is known beyond these relatively direct extensions. 

In particular, Gr\"otschel, Lov\'asz, and Schrijver~\cite{groetschel_1981_ellipsoid} (see also~\cite{groetschel_1984_corrigendum,groetschel_1993_geometric}) showed that SFM can be solved efficiently over all odd or even sets. This extended a well-known earlier result by Padberg and Rao~\cite{padberg_1982_odd}, showing that minimum odd cuts can be found efficiently, and also a later extension by Barahona and Conforti~\cite{barahona_1987_construction} to even cuts.
More precisely, Gr\"otschel, Lov\'asz, and Schrijver~\cite{groetschel_1984_corrigendum} show that SFM can be solved efficiently over any family of sets $\mathcal{F}\subseteq 2^N$ that is a \emph{triple family}, which they define as follows:
For any $A,B \subseteq N$, if three of the four sets $A,B,A\cup B$, and $A\cap B$ are not in $\mathcal{F}$, then none of the four sets are in $\mathcal{F}$. One can easily check that all even or odd sets indeed form a triple family.
The most general constraint family under which SFM is known to be efficiently solvable was introduced by Goemans and Ramakrishnan~\cite{goemans_1995_minimizing}. They showed that SFM can be efficiently solved over a generalization of triple families, which they called \emph{parity family}. A set family $\mathcal{F}\subseteq 2^N$ is a parity family if for any pair of sets $A,B\subseteq N$ with $A\notin\mathcal{F}$ and $B\not\in\mathcal{F}$, either both of $A\cup B$ and $A\cap B$ are in $\mathcal{F}$, or none of the two.

The difficulty in identifying relevant constraint classes under which SFM can be done efficiently is partially explained by the fact that SFM can quickly become very hard, even under constraint types for which other problems, like submodular maximization, can still be solved approximately. More precisely, Svitkina and Fleischer~\cite{svitkina_2011_submodular} showed that even with a single cardinality lower bound, monotone SFM is impossible to approximate in the oracle model up to a factor $o(\sqrt{\sfrac{n}{\log n}})$, where $n\coloneqq |N|$ is the size of the ground set.

\medskip

The goal of this work is to present a new natural constraint class under which efficient submodular minimization is possible, and which is motivated by recent progress in linear integer programming with bounded subdeterminants, and by recent open questions related to binary matroids. More precisely, we consider the following natural generalization of parity-constrained submodular minimization.

\begin{mdframed}[leftmargin=1cm,rightmargin=1cm]
{\bf Congruency-Constrained Submodular Minimization (CCSM):}
Let $f\colon\mathcal{L} \rightarrow \mathbb{Z}$ be a submodular function defined on a lattice $\mathcal{L}\subseteq 2^N$, and let $m \in \mathbb{Z}_{>0}$, $r\in \{0,\ldots, m-1\}$. The task is to find a minimizer of
\begin{equation}\label{eq:CCSM}\tag{CCSM}
\min\{f(S) \mid S\in \mathcal{L}, \; |S| \equiv r \pmod*{m}\}\enspace.
\end{equation}
\end{mdframed}
We call $m$ the \emph{modulus} of the \eqref{eq:CCSM} problem. Moreover, we highlight that $N$ is a finite ground set throughout this paper.
Notice that the case $m=2$ captures odd/even submodular minimization, and thus in particular the odd cut problem. More generally, one can observe that also the $T$-cut problem, which only considers cuts with an odd number of vertices within a vertex set $T$, can easily be cast as~\eqref{eq:CCSM}.%
\footnote{
Indeed, we can observe that any problem of the form $\min\{f(S) \mid S\in \mathcal{L}, \; |S\cap T| \equiv r \pmod*{2}\}$, for a submodular function $f\colon\mathcal{L} \rightarrow \mathbb{Z}$ with $\mathcal{L}\subseteq 2^N$ and a given set $T\subseteq N$, can be cast as a \eqref{eq:CCSM} problem with respect to an auxiliary submodular function $g$ over a lattice $\mathcal{L}'$ as follows. For every $x\in N\setminus T$, introduce a new element $x'$, let $N'\coloneqq N\cup \{x'\mid x\in N\setminus T\}$, and let $\mathcal{L}'\coloneqq \{S\subseteq N' \mid S\cap N \in \mathcal{L} \text{ and }|S\cap\{x,x'\}|\neq 1\ \forall x\in N\setminus T\}$. Define $g\colon\mathcal{L}'\to\mathbb{Z}$ by $g(S)=f(S\cap N)$ for all $S\subseteq N'$. Then, for any $S^*\in\argmin\{g(S)\mid S\in\mathcal{L}',\; |S|\equiv r\pmod*{2}\}$, we can observe that $S^*\cap N$ solves the original problem. The same construction can be applied for moduli $m$ different from $2$ by introducing $m-1$ copies for each element in $N\setminus T$.%
}
Apart from naturally extending known SFM settings, our study of~\eqref{eq:CCSM} is motivated by an open question in integer programming, namely whether integer linear programs (ILPs) with constraint matrices having constantly bounded subdeterminants can be solved efficiently. 
More precisely, it was recently shown in \cite{artmann_2017_strongly} that bimodular ILPs can be solved efficiently, which are problems of the form $\max\{c^T x \mid Ax \leq b, x\in \mathbb{Z}^n\}$, where $A$ has full column rank and each $n\times n$ submatrix of $A$ has a determinant within $\{-2,-1,0,1,2\}$. This result implies that any ILP such that all subdeterminants of $A$ are within $\{-2,-1,0,1,2\}$ can be solved efficiently (see~\cite{artmann_2017_strongly} for more details), thus extending the well-known fact that ILPs with totally unimodular constraint matrices are efficiently solvable. However, whether ILPs with larger subdeterminants can still be solved efficiently seems to be a question beyond current techniques.
Interestingly, a key algorithmic tool used in~\cite{artmann_2017_strongly} to show that bimodular ILPs are efficiently solvable is efficient odd submodular minimization, or at least efficient algorithms to find minimum directed $T$-cuts, since the submodular minimization problems appearing in~\cite{artmann_2017_strongly} can be reformulated as directed $T$-cut problems. Conversely, a directed $T$-cut problem can naturally be modeled as a bimodular ILP. ILPs with subdeterminants up to $m$ include a natural extension of the directed $T$-cut problem, namely the problem of finding a cut of smallest value among all cuts of cardinality $r \bmod{m}$. This is clearly a special case of~\eqref{eq:CCSM}, by choosing $f$ to be the directed cut function.
Hence, to make progress on the question of ILPs with bounded subdeterminants, one needs to be able to solve~\eqref{eq:CCSM} for $f$ being an arbitrary directed cut function. Furthermore, due to the approach presented in~\cite{artmann_2017_strongly}, there is hope that this subproblem may be an important building block for finding an efficient procedure to solve ILPs with bounded subdeterminants.

Moreover, we consider the following generalized version of~\eqref{eq:CCSM}, which nicely highlights the versatility of our approach and captures several open problems raised by Geelen and Kapadia~\cite{geelen_2017_computing} in the context of computing the girth and cogirth of perturbed graphic matroids.
In the definition below, as well as later in the paper, we use the shorthand $[k]\coloneqq \{1,\ldots,k\}$.

\begin{mdframed}[leftmargin=1cm,rightmargin=1cm]
{\bf Generalized Congruency-Constrained Submodular Minimization (GCCSM):}
\sloppy Let $f\colon\mathcal{L} \rightarrow \mathbb{Z}$ be a submodular function defined on a lattice $\mathcal{L}\subseteq 2^N$, and let $m \in \mathbb{Z}_{>0}$. Moreover, let $k\in\mathbb{Z}_{>0}$, $S_1,\ldots, S_k\subseteq N$ and $r_1,\ldots, r_k\in \{0,\ldots, m-1\}$. The task is to find a minimizer of 
\begin{equation}\label{eq:GCCSM}
\min\{f(S) \mid S\in \mathcal{L}, \; |S\cap S_i| \equiv r_i \pmod*{m}\;\;\forall i\in [k]\}\enspace.
\tag{GCCSM}
\end{equation}
\end{mdframed}
In particular,~\eqref{eq:GCCSM} captures the $t$-Set Even-Cut Problem and $t$-Set Odd-Cut Problem defined in~\cite{geelen_2017_computing}. There, one is given a constant $t$, an undirected graph $G=(V,E)$, and sets $T_1,\ldots, T_t\subseteq V$. The task is to find a cut $S\subseteq V$ with a minimum number of edges $|\delta(S)|$ among all cuts whose intersections with the sets $T_i$ are all even or all odd, respectively.
Geelen and Kapadia identified the $t$-Set Even-Cut Problem as a special case of the so-called $t$-Dimensional Even-Cut Problem. While the latter is key to their algorithm for computing the cogirth of perturbed graphic matroids, they consider the $t$-Set Even-Cut problem as a purer form of the problem, which they believe to be of independent interest, as well as the natural variation of the $t$-Set Odd-Cut problem.
Geelen and Kapadia present a randomized algorithm for the $t$-Set Even-Cut Problem, based on an adaptation of Karger's contraction algorithm~\cite{karger_1993_global,karger_1996_new}, and they raise the following open questions which we address through our work:
\begin{enumerate}[label=(\roman*)]
\item They ask about a deterministic procedure for the $t$-Set Even-Cut problem, which they mention as one of the main shortcomings of their approach. As noted in~\cite{geelen_2017_computing}, Conforti and Rao~\cite{conforti_1987_some} found an efficient deterministic algorithm for the $1$-Set Even-Cut Problem. However, even for the $2$-Set version, no deterministic procedure is known.

\item They raise the question about the complexity of the Odd-Cut problem, stating that the method of Padberg and Rao~\cite{padberg_1982_odd} for finding an odd cut extends to the $2$-Set Odd-Cut setting; however, even for the $3$-Set Odd-Cut problem, the complexity remains open.
\end{enumerate}

The main technical contribution of this paper is to introduce a new approach based on techniques from Combinatorics and Number Theory to analyze a natural algorithm for~\eqref{eq:CCSM} and~\eqref{eq:GCCSM}.

\subsection{Our results}

We start by stating the implications of our techniques on~\eqref{eq:CCSM} and~\eqref{eq:GCCSM}, and provide an overview of the techniques in Section~\ref{subsec:overviewTech}.
Our main result for~\eqref{eq:CCSM} is the following.
\begin{theorem}\label{thm:mainCCSM}
For any $m\in \mathbb{Z}_{>0}$ that is a prime power, \eqref{eq:CCSM} can be solved in time $n^{2m+O(1)}$.
\end{theorem}
Hence, we can efficiently solve~\eqref{eq:CCSM} for any modulus $m$ that is a prime power bounded by a constant. Notice that an upper bound on $m$ is required to obtain an efficient algorithm. Indeed, in particular if $m=n\coloneqq |N|$, the congruency constraint simply models a cardinality constraint. However, as mentioned in the introduction, SFM subject to a cardinality constraint is impossible to approximate up to any factor $o(\sqrt{\sfrac{n}{\log n}})$ in the oracle model. It is not hard to observe that this implies that even for any $\epsilon >0$, \eqref{eq:CCSM} with modulus $m=\Omega(n^\epsilon)$ cannot be solved exactly in polynomial time.

Our key contribution, which leads to Theorem~\ref{thm:mainCCSM}, is a connection of the correctness of a natural procedure, which we introduce in Section~\ref{subsec:overviewTech}, and the nonexistence of certain set systems. To disprove the existence of such set systems, we employ tools from Combinatorics and Number Theory, in particular Fermat's Little Theorem.
A main advantage of our techniques is that they are very versatile, and allow in particular for an adaptation to~\eqref{eq:GCCSM}, leading to the following result.

\begin{theorem}\label{thm:mainGCCSM}
For any $m\in \mathbb{Z}_{>0}$ that is a prime power, \eqref{eq:GCCSM} can be solved in time $n^{2km+O(1)}$.
\end{theorem}

Notice that Theorem~\ref{thm:mainGCCSM} solves the two open questions by Geelen and Kapadia~\cite{geelen_2017_computing} mentioned in the introduction.
Moreover, we want to highlight that our algorithms for solving~\eqref{eq:CCSM} and~\eqref{eq:GCCSM} consist of repeatedly solving unconstrained submodular function minimization problems, namely at most $n^{2(m-1)}$ many for~\eqref{eq:CCSM} and $n^{2k(m-1)}$ many for~\eqref{eq:GCCSM}. Using a strongly polynomial algorithm for submodular function minimization, the running time guarantees of Theorems~\ref{thm:mainCCSM} and~\ref{thm:mainGCCSM} are achieved.

\subsection{Overview of main steps of our technique}
\label{subsec:overviewTech}

\newcommand{\refEnumD}[1][d]{%
\hyperlink{alg:enumTarget}{$\mathrm{Enum}(#1)$}\xspace%
}%

We start by stating a natural algorithm, \refEnumD highlighted below, that we use to derive both of our main results, Theorems~\ref{thm:mainCCSM} and~\ref{thm:mainGCCSM}. 
Our algorithm is parameterized by an integer $d\in \mathbb{Z}_{> 0}$, which we call the \emph{depth} of the algorithm. Its input is a value oracle for a submodular function $f\colon\mathcal{L} \rightarrow \mathbb{Z}$ defined on a lattice $\mathcal{L}\subseteq 2^N$, and a family $\mathcal{F}\subseteq 2^N$, capturing additional constraints we want to satisfy. In particular, for~\eqref{eq:CCSM} we have $\mathcal{F}=\{S\subseteq N \mid |S|\equiv r \pmod{m}\}$, and for~\eqref{eq:GCCSM}, the set $\mathcal{F}$ is given by $\mathcal{F}=\{S\subseteq N \mid |S\cap S_i|\equiv r_i \pmod{m} \;\forall i\in[k]\}$. We assume that $\mathcal{F}$ is given by a membership oracle, which can be queried for any set $S\subseteq N$, and returns whether $S\in \mathcal{F}$.

{%
\renewcommand{\thealgocf}{}
\begin{algorithm2e}[!h]
\SetAlgorithmName{$\bm{\mathrm{Enum}(d)}$\labeltarget{alg:enumTarget}}{}

\begin{enumerate}[rightmargin=0.35em,leftmargin=1.2em,partopsep=-0.5em,itemsep=5pt]
\item\label{algitem:enum} For all $A,B \subseteq N$ with $|A|,|B|\leq d$ and $A\cap B=\emptyset$, compute a minimal minimizer of $f$ over the lattice
\begin{equation*}
\mathcal{L}_{AB} \coloneqq \{S\in \mathcal{L} \mid A\subseteq S \subseteq N\setminus B\}\enspace.
\vspace*{-0.3em}
\end{equation*}
Let $\mathcal{S}$ be the family of all computed minimal minimizers for all pairs of $A$ and $B$.

\item Return a set $S\in \mathcal{S}$ of minimum value among all sets in $\mathcal{S}\cap\mathcal{F}$.
\end{enumerate}

\caption{Enumeration algorithm of depth $d$ for submodular minimization over $\mathcal{F}$}\label{alg:EnumD}
\end{algorithm2e}
}

The algorithm is a natural extension of a procedure suggested in~\cite{goemans_1995_minimizing}, which corresponds to \refEnumD[1].
In step~\ref{algitem:enum}, we repeatedly solve unconstrained submodular minimization problems for minimal minimizers. To this end, one can observe that many submodular function minimization algorithms do actually return minimal minimizers. Alternatively, for integer-valued submodular functions, we can observe that a set is a minimal minimizer of $f$ if and only if it is a minimizer of the submodular function $g$ given by $g(S)=(n+1)f(S)+|S|$. Hence, it suffices to find any minimizer of $g$ to obtain a minimial minimizer of $f$.
Notice that \refEnumD is clearly a polynomial time algorithm for any constant depth $d$. However, depending on the structure of the constraint set $\mathcal{F}$, and the choice of $d$, the above algorithm may fail to return a set in $\mathcal{F}$ with minimum submodular value. In particular, it may even happen that no feasible solution is found, i.e., $\mathcal{S}\cap \mathcal{F} = \emptyset$.
In the following, we show the main steps that we used to derive our main result for~\eqref{eq:CCSM}, i.e., Theorem~\ref{thm:mainCCSM}. In Section~\ref{sec:GCCSM}, we show how to extend the results to~\eqref{eq:GCCSM}.

To analyze the correctness of \refEnumD for \eqref{eq:CCSM}, we show that if \refEnumD fails to return an optimal solution to \eqref{eq:CCSM}, then this implies the existence of a set system with the following properties.
For brevity, we call a set system satisfying these properties an $(m,d)$-system.

\begin{definition}[$(m,d)$-system (on $N$)]
Let $N$ be a finite ground set, and let $m,d\in \mathbb{Z}_{>0}$. We say that a set system $\mathcal{H}\subseteq 2^N$ is an $(m,d)$-system (on $N$) if
\begin{enumerate}[label=(\roman*),itemsep=-0.2em,topsep=-0.2em]
\item\label{item:MDIntClosed} $\mathcal{H}$ is closed under intersection, i.e., $H_1\cap H_2\in \mathcal{H} \;\;\forall H_1,H_2\in \mathcal{H}$,

\item\label{item:MDDiffParity} $|H| \not\equiv |N| \pmod{m} \;\;\forall H\in \mathcal{H}$, and

\item\label{item:MDCoverage} for any $S\subseteq N$ with $|S|\leq d$, there is a set $H\in \mathcal{H}$ with $S\subseteq H$.
\end{enumerate}
\end{definition}

Note that in particular, we require property~\ref{item:MDIntClosed} also for disjoint sets: If there are $H_1,H_2\in\mathcal{H}$ with $H_1\cap H_2=\emptyset$, then $\emptyset\in\mathcal{H}$. On the other hand, if $\emptyset\not\in\mathcal{H}$, we can conclude that all sets have at least one element in common. Also observe that property~\ref{item:MDCoverage} implies that the sets of an $(m,d)$-system $\mathcal{H}$ cover the ground set, i.e., we always have $N=\bigcup_{H\in\mathcal{H}}H$.
The following theorem formalizes a crucial link between nonexistence of $(m,d)$-systems and correctness of \refEnumD, and reduces the correctness of \refEnumD to a purely combinatorial question. Here (and throughout the rest of this paper), \emph{nonexistence of $(m,d)$-systems} without explicit reference to a ground set is to be understood to hold for any ground set, i.e., no matter what finite ground set $N$ is chosen, there does not exist an $(m,d)$-system on $N$.

\begin{theorem}\label{thm:EnumDGoodIfNoBadSys}
Let $m,d\in \mathbb{Z}_{>0}$. If no $(m,d)$-system exists, then \refEnumD returns an optimal solution to any \eqref{eq:CCSM} problem with modulus $m$.%
\end{theorem}

Notice that Theorem~\ref{thm:EnumDGoodIfNoBadSys} does not depend on the lattice $\mathcal{L}$ underlying the \eqref{eq:CCSM} problem. For specific lattices $\mathcal{L}\subseteq 2^N$, the above conditions can be slightly weakened. In particular, it suffices to consider a weaker definition of $(m,d)$-systems, where $\mathcal{H}$ needs to be a subfamily of $\mathcal{L}$. Section~\ref{sec:reduceToSetSys} provides further details on this. However, for the congruency constraints we consider, we do not need the weaker requirements for specific lattices, and we thus decided to avoid these details here in the interest of simplifying the presentation.

Finally, our approach is completed by deriving the following result, which completes the last step of our proof, and, together with Theorem~\ref{thm:EnumDGoodIfNoBadSys}, implies Theorem~\ref{thm:mainCCSM}.
\begin{theorem}\label{thm:noMM-1Sys}
For $m\in \mathbb{Z}_{>0}$ being a prime power, there is no $(m,m-1)$-system.
\end{theorem}

Moreover, we want to mention that Gopi~\cite{gopi2017systems}, after hearing a presentation of this paper, found an elegant proof showing that for $m$ not being a prime power, there do exist $(m,m-1)$-systems. This shows an interesting discrepancy between prime power moduli and non-prime power moduli, and it suggests that an extension of our techniques to the latter case requires new ideas.

\subsection{Organization of the paper}

In Section~\ref{sec:reduceToSetSys}, we show how the correctness of \refEnumD can be reduced to the nonexistence of $(m,d)$-systems, thus proving Theorem~\ref{thm:EnumDGoodIfNoBadSys}.
Section~\ref{sec:noBadSetSystem} shows Theorem~\ref{thm:noMM-1Sys}, the nonexistence of $(m,m-1)$-systems for $m$ being a prime power. The techniques presented in Section~\ref{sec:noBadSetSystem} comprise a general framework based on results from Combinatorics and Number Theory to disprove existence of certain types of set systems. In Section~\ref{sec:GCCSM}, we show how these techniques can be extended to~\eqref{eq:GCCSM}, thus implying our main result for~\eqref{eq:GCCSM}, Theorem~\ref{thm:mainGCCSM}.
Section~\ref{sec:barriers} identifies a combinatorial barrier to extending our proof techniques beyond $m$ being a prime power.
Section~\ref{sec:existenceMMm2Systems} shows that our choice of the depth $d$ of \refEnumD is smallest possible for the problems we consider.

\section{Reducing correctness of \refEnumD to properties of set systems}
\label{sec:reduceToSetSys}

The main goal of this section is to prove Theorem~\ref{thm:EnumDGoodIfNoBadSys}.
In fact, we show a slight strengthening, which allows us to derive results for \eqref{eq:GCCSM}, and may lead to further applications for constraints beyond congruency constraints. For this we generalize the notion of $(m,d)$-system to the notion of an $(\mathcal{F},d)$-system, where the role of all sets of cardinality $r\bmod{m}$ is replaced by a general constraint family $\mathcal{F}\subseteq 2^N$ on $N$. Moreover, we will be explicit about the underlying lattice, which leads to stronger statements that may be helpful for extending our results to further contexts.

\begin{definition}[$(\mathcal{F},d)$-system]
Let $\mathcal{L}\subseteq 2^N$ be a lattice, $\mathcal{F}\subseteq \mathcal{L}$, and let $d\in \mathbb{Z}_{>0}$. A family $\mathcal{H}\subseteq \mathcal{L}$ is called an $(\mathcal{F},d)$-system if the following holds, where $Q\coloneqq \bigcup_{H\in \mathcal{H}}H$:
\begin{enumerate}[label=(\roman*),itemsep=-0.2em,topsep=-0.2em]
\item\label{item:FDQin} $Q\in \mathcal{F}$,
\item\label{item:FDIntClosed} $\mathcal{H}$ is closed under intersection,
\item\label{item:FDSetsNotFeasible} $H\not\in \mathcal{F} \;\;\forall H\in \mathcal{H}$, and
\item\label{item:FDCoverage} for any $S\subseteq Q$ with $|S|\leq d$, there is a set $H\in \mathcal{H}$ with $S\subseteq H$.
\end{enumerate}
\end{definition}

Using the notion of $(\mathcal{F},d)$-systems, we can now define the following strengthening of Theorem~\ref{thm:EnumDGoodIfNoBadSys}, where for any set family $\mathcal{F}\subseteq \mathcal{L}$ defined on a lattice $\mathcal{L}$, we denote by $\comp(\mathcal{F})$ the complement family, i.e., $\comp(\mathcal{F})\coloneqq\{N\setminus F \mid F\in \mathcal{F}\}$, which we will interpret as a subfamily of the lattice $\comp(\mathcal{L})$.

\begin{theorem}\label{thm:EnumDGoodIfNoBadSysGen}
Let $\mathcal{L}\subseteq 2^N$ and $\mathcal{F}\subseteq \mathcal{L}$, and let $d\in \mathbb{Z}_{>0}$. If no $(\mathcal{F},d)$-system and no $(\comp(\mathcal{F}),d)$-system exists, then \refEnumD returns an optimal solution to any submodular function minimization problem over $\mathcal{F}$.
\end{theorem}

We start by observing that Theorem~\ref{thm:EnumDGoodIfNoBadSysGen} indeed implies Theorem~\ref{thm:EnumDGoodIfNoBadSys}.

\begin{proof}[Proof of Theorem~\ref{thm:EnumDGoodIfNoBadSys}]
Consider a \eqref{eq:CCSM} problem $\min\{f(S) \mid S\in \mathcal{L}, |S|\equiv r \pmod{m}\}$. Hence, the set family $\mathcal{F}$ over which we want to minimize the function $f$ is given by
\begin{equation*}
\mathcal{F} = \{S\in \mathcal{L} \mid |S| \equiv r \pmod*{m}\}\enspace,
\end{equation*}
and its complement family is therefore
\begin{align*}
\comp(\mathcal{F}) = \{S\in \comp(\mathcal{L}) \mid |N\setminus S| \equiv r \pmod*{m}\}
 = \{S\in \comp(\mathcal{L}) \mid |S| \equiv |N| -r \pmod*{m}\}\enspace.
\end{align*}
The proof now follows by observing that any $(\mathcal{F},d)$-system or $(\comp(\mathcal{F}),d)$-system is also an $(m,d)$-system on a potentially different ground set.
Indeed, consider an $(\mathcal{F},d)$-system $\mathcal{H}$, and let $Q=\bigcup_{H\in \mathcal{H}}H$. (The case of a $(\comp(\mathcal{F}),d)$-system is analogous.) Then $\mathcal{H}$ is an $(m,d)$-system on $Q$ because properties~\ref{item:FDIntClosed} and~\ref{item:FDCoverage} of the definition of an $(\mathcal{F},d)$-system correspond to properties~\ref{item:MDIntClosed} and~\ref{item:MDCoverage} of an $(m,d)$-system on $Q$, respectively; moreover, properties~\ref{item:FDQin} and~\ref{item:FDSetsNotFeasible} of an $(\mathcal{F},d)$-system imply property~\ref{item:MDDiffParity} of an $(m,d)$-system.
\end{proof}

It remains to prove Theorem~\ref{thm:EnumDGoodIfNoBadSysGen}.

\subsection{Proof of Theorem~\ref{thm:EnumDGoodIfNoBadSysGen}}

We start by stating a key property of set systems $\mathcal{F}\subseteq \mathcal{L}$ that is crucial in our analysis to show that \refEnumD returns an optimal solution. This is an extension of a property used in~\cite{goemans_1995_minimizing} for parity constraints.

\begin{definition}[$d$-good set system]
Let $\mathcal{L}\subseteq 2^N$ be a lattice and $\mathcal{F}\subseteq \mathcal{L}$. We say that the tuple $(\mathcal{F}, \mathcal{L})$ is \emph{$d$-good}---or simply that $\mathcal{F}$ is \emph{$d$-good} if $\mathcal{L}$ is clear from context---if for any submodular function $f\colon\mathcal{L} \rightarrow \mathbb{Z}$, and any minimizer $S^*$ of $\min\{f(S) \mid S\in \mathcal{F}\}$, there exists a set $A\subseteq S^*$ with $|A|\leq d$ satisfying 
\begin{equation*}
f(S) \geq f(S^*) \quad \forall S\in \mathcal{L} \text{ with } A\subseteq S \subseteq S^*\enspace.
\end{equation*}
\end{definition}

We now prove Theorem~\ref{thm:EnumDGoodIfNoBadSysGen} in two steps. First, we show that if a set system $\mathcal{F}$ and its complement family $\comp(\mathcal{F})$ are $d$-good, then our algorithm will return an optimal solution.

\begin{lemma}\label{lem:dGoodToOpt}
Let $\mathcal{L}\subseteq 2^N$ be a lattice, $\mathcal{F}\subseteq \mathcal{L}$, and $d\in \mathbb{Z}_{>0}$. If $(\mathcal{F},\mathcal{L})$ and $(\comp(\mathcal{F}),\comp(\mathcal{L}))$ are both $d$-good, then \refEnumD returns an optimal solution to any submodular minimization problem on $\mathcal{F}$.
\end{lemma}

Conversely, if a constraint set $\mathcal{F}$ is not $d$-good, then we can derive the existence of an $(\mathcal{F},d)$-system out of it as shown by the following lemma, which, together with Lemma~\ref{lem:dGoodToOpt}, immediately implies Theorem~\ref{thm:EnumDGoodIfNoBadSysGen}, as desired.

\begin{lemma}\label{lem:notDGoodToSys}
Let $\mathcal{L}\subseteq 2^N$ be a lattice, $\mathcal{F}\subseteq \mathcal{L}$, and $d\in \mathbb{Z}_{>0}$. If  $(\mathcal{F},\mathcal{L})$ is not $d$-good, then there exists an $(\mathcal{F},d)$-system.
\end{lemma}

The proof strategies for Lemmas~\ref{lem:dGoodToOpt} and~\ref{lem:notDGoodToSys} are heavily inspired by an approach presented in~\cite{goemans_1995_minimizing} for parity families. We remark that the proof of Lemma~\ref{lem:dGoodToOpt} strengthens the proof approach presented in~\cite{goemans_1995_minimizing}, which allows us to use simpler requirements for the definition of a $d$-good system than what would have been necessary by following the proof approach in~\cite{goemans_1995_minimizing} more closely.

To prove Lemma~\ref{lem:dGoodToOpt}, we show that under the assumption that $(\mathcal{F},\mathcal{L})$ and $(\comp(\mathcal{F}),\comp(\mathcal{L}))$ are both $d$-good, \refEnumD[d] returns a set with function value equal to the function value of a minimal optimal solution. As the following lemma shows, arguing about minimal (with respect to inclusion) optimal solutions allows us to obtain a stronger result from $(\mathcal{F},\mathcal{L})$ being a $d$-good set system.

\begin{lemma}\label{lem:dGoodForMinimal}
Let $(\mathcal{F},\mathcal{L})$ be a $d$-good set system. Then, for any submodular function $f\colon \mathcal{L} \rightarrow \mathbb{Z}$, and any minimal minimizer $S^*$ of $\min\{f(S) \mid S\in \mathcal{F}\}$, there exists a set $A\subseteq S^*$ with $|A|\leq d$ satisfying 
\begin{equation*}
f(S) > f(S^*) \quad \forall S\in \mathcal{L} \text{ with } A\subseteq S \subsetneq S^*\enspace.
\end{equation*}
\end{lemma}

\begin{proof}
Fix a submodular function $f\colon \mathcal{L} \rightarrow \mathbb{Z}$, and a minimal minimizer $S^*$ of $\min\{f(S) \mid S\in \mathcal{F}\}$. Consider the function $g\colon \mathcal{L} \rightarrow \mathbb{Z}$ given by
\begin{equation*}
g(S) = |N|f(S) + |N||S\setminus S^*| + |S| \quad \text{for all } S\in\mathcal{L}\enspace.
\end{equation*}
The function $g$ is submodular because it is a conic combination of the three submodular functions $S\mapsto f(S)$, $S\mapsto |S\setminus S^*|$, and $S\mapsto |S|$. Moreover, we claim that $S^*$ is a minimizer of $\min\{g(S) \mid S\in\mathcal{F}\}$.
Indeed, by definition of $S^*$, we have $f(S) \geq f(S^*)$ for all $S\in\mathcal{F}$. If $f(S)\geq f(S^*)+1$, we get
\begin{align*}
g(S) \geq |N| f(S) \geq |N| (f(S^*)+1) \geq |N| f(S^*) + |S^*| = g(S^*)\enspace. 
\end{align*}
If, in the other case, $f(S) = f(S^*)$, then minimality of $S^*$ implies that $S\setminus S^* \neq \emptyset$, hence $|S\setminus S^*| \geq 1$, so
\begin{align*}
g(S) \geq |N| f(S) + |N| |S\setminus S^*| \geq |N| f(S^*) + |N| \geq |N| f(S^*) + |S^*| = g(S^*)\enspace. 
\end{align*}

Applying the property that $(\mathcal{F},\mathcal{L})$ is $d$-good to the submodular function $g$ and the minimizer $S^*$ of $\min\{g(S) \mid S\in\mathcal{F}\}$, we obtain that there exists a set $A\subseteq S^*$ with $|A|\leq d$ satisfying
\begin{equation*}
g(S)\geq g(S^*) \quad \forall S\in \mathcal{L} \text{ with } A\subseteq S \subseteq S^*\enspace.
\end{equation*}
To conclude, it suffices to see that for all $S\in\mathcal{L}$ with $S\subsetneq S^*$, the inequality $g(S) \geq g(S^*)$ implies $f(S) > f(S^*)$. Note that $S\subsetneq S^*$ implies $|S\setminus S^*|=0$, so the inequality $g(S) \geq g(S^*)$ can be rewritten as
\begin{equation*}
|N| f(S) + |S| \geq |N| f(S^*) + |S^*|\enspace.
\end{equation*}
The assumption $S\subsetneq S^*$ also implies $|S| < |S^*|$, hence, from the last inequality, we conclude $f(S)>f(S^*)$.
\end{proof}

With the above strengthening, we are ready to prove Lemma~\ref{lem:dGoodToOpt}.

\begin{proof}[Proof of Lemma~\ref{lem:dGoodToOpt}]
Let $f\colon\mathcal{L} \rightarrow \mathbb{Z}$ be a submodular function and let $S^*$ be a minimal minimizer of $\min\{f(S) \mid S\in\mathcal{F}\}$. Using that $(\mathcal{F},\mathcal{L})$ is $d$-good and applying Lemma~\ref{lem:dGoodForMinimal}, we obtain existence of a set $A\subseteq S^*$ with $|A|\leq d$ satisfying
\begin{equation}\label{eq:innerIneq}
f(S) > f(S^*) \quad \forall S\in \mathcal{L} \text{ with } A\subseteq S \subsetneq S^*\enspace.
\end{equation}
Note that the function $g\colon \comp(\mathcal{L})\to\mathbb{Z}$ given by $g(S)=f(N\setminus S)$ is submodular, and $N\setminus S^*$ is a minimizer of $\min\{ g(S) \mid S\in\comp(\mathcal{F})\}$. So using the assumption that $(\comp(\mathcal{F}),\comp(\mathcal{L}))$ is $d$-good, we obtain existence of a set $B\subseteq N\setminus S^*$ with $|B|\leq d$ satisfying
\begin{equation*}
g(S) \geq g(N\setminus S^*) \quad \forall S\in \comp(\mathcal{L}) \text{ with } B\subseteq S \subseteq N\setminus S^*\enspace.
\end{equation*}
Rewriting the above in terms of $f$ and replacing $S$ by $N\setminus S$, we get
\begin{equation}\label{eq:outerIneq}
f(S) \geq f(S^*) \quad \forall S \in \mathcal{L} \text{ with } S^* \subseteq S \subseteq N\setminus B\enspace.
\end{equation}

Let $T$ be a minimal minimizer of $f$ over the lattice $\mathcal{L}_{AB}=\{S\in\mathcal{L} \mid A\subseteq S\subseteq N\setminus B\}$. Note that sets of this type are found in the first step of \refEnumD[d] when considering the sets $A$ and $B$ given above. We claim that in fact $T=S^*$, proving that the minimizer $S^*$ is found in the first step of \refEnumD[d]. Consequently, the set returned by \refEnumD[d] is a set of optimal value $f(S^*)$, which is what we wanted to prove.

It remains to see that $T=S^*$. As $S^*\in\mathcal{L}_{AB}$, we have $f(S^*)\geq f(T)$. Together with submodularity of $f$, we get
\begin{equation*}
2f(S^*) \geq f(S^*) + f(T) \geq f(S^*\cap T) + f(S^* \cup T)\enspace.
\end{equation*}
If $S^*\cap T \subsetneq S^*$,~\eqref{eq:innerIneq} implies $f(S^*\cap T)>f(S^*)$. Moreover,~\eqref{eq:outerIneq} implies $f(S^*\cup T)\geq f(S^*)$. Together, we obtain a contradiction to the previous inequality. Consequently, we have $S^*\cap T = S^*$ or, in other words, $S^*\subseteq T$. Minimality of $T$ implies $S^* = T$, as desired.
\end{proof}

Finally, we prove Lemma~\ref{lem:notDGoodToSys}, which is the last missing piece in our proof of Theorem~\ref{thm:EnumDGoodIfNoBadSysGen}.

\begin{proof}[Proof of Lemma~\ref{lem:notDGoodToSys}]
Assume that $(\mathcal{F}, \mathcal{L})$ is not $d$-good. Hence, there is a submodular function $f\colon\mathcal{L}\to\mathbb{Z}$ and a minimizer $S^*$ of $\min\{f(S) \mid S\in \mathcal{F}\}$ such that for any set $A\subseteq S$ with $|A|\leq d$, there is a set $S_A\in \mathcal{L}$ with $A\subseteq S_A \subseteq S^*$ satisfying $f(S_A) < f(S^*)$. Among all such sets $S_A$, we choose one that is maximal (inclusion-wise).
Let $\mathcal{H}\subseteq \mathcal{L}$ be the family of all sets that can be obtained as intersections of the sets $\{S_A\}_{A\subseteq S^*, |A|\leq d}$, where we include the sets $S_A$ themselves also in the family $\mathcal{H}$. We claim that $\mathcal{H}$ is an $(\mathcal{F},d)$-system. 

Clearly, $\mathcal{H}\subseteq \mathcal{L}$, because each set $S_A$ satisfies $S_A \in \mathcal{L}$ and the lattice $\mathcal{L}$ is closed under intersection. Moreover, we have $Q = \bigcup_{H \in \mathcal{H}} H = S^*$ because each set in $\mathcal{H}$ is contained in $S^*$, and for each element $e\in S^*$, the set $S_{\{e\}}\in \mathcal{H}$ contains $e$.
Property~\ref{item:FDQin} of an $(\mathcal{F},d)$-system follows from $Q=S^*\in \mathcal{F}$. Moreover,~\ref{item:FDIntClosed} holds because $\mathcal{H}$ is intersection-closed by construction. Property~\ref{item:FDCoverage} is fulfilled because for each $A\subseteq Q$ with $|A|\leq d$, the set $S_A\in \mathcal{H}$ fulfills $A\subseteq S_A$. It remains to show that $\mathcal{H}$ fulfills property~\ref{item:FDSetsNotFeasible} of an $(\mathcal{F},d)$-system, i.e., that each set $H\in \mathcal{H}$ satisfies $H\not\in \mathcal{F}$.
Recall that each set $H\in \mathcal{H}$ can be written as
\begin{equation}\label{eq:HAsIntersection}
H = \bigcap_{i=1}^k S_{A_i}\enspace,
\end{equation}
where $k\in \mathbb{Z}_{\geq 1}$, and $A_1, \ldots, A_k \in \mathcal{L}$ with $A_i \subseteq S^*$ and $|A_i| \leq d$ for $i\in [k]$. We show that $f(H) < f(S^*)$ by induction on $k$. Notice that this implies $H\not\in \mathcal{F}$ because $S^*$ is a minimizer of $\min\{f(S) \mid S \in \mathcal{F}\}$, and hence, no other set in $\mathcal{F}$ can have a smaller $f$-value.

The case $k=1$ corresponds to sets $H=S_A$, where $A\in \mathcal{L}$, $A\subseteq S^*$, and $|A|\leq d$. By our choice of the sets $S_A$, we have $f(S_A) < f(S^*)$ for these sets.
Now consider a set $H$ as described in~\eqref{eq:HAsIntersection} for $k\geq 2$, and assume that for any set $H'$ that can be described as the intersection of at most $k-1$ sets $S_{A_i}$, it holds that $f(H') < f(S^*)$. Let $H' = \bigcap_{i=1}^{k-1} S_{A_i}$, and hence, $H= H' \cap S_{A_k}$. By submodularity of $f$ we have
\begin{equation}\label{eq:HAsIntSAk}
f(H') + f(S_{A_k}) \geq f(H'\cup S_{A_k}) + f(H)\enspace.
\end{equation}
By definition, $S_{A_k}$ is a maximal subset of $S^*$ containing $A_k$ and satisfying $f(S_{A_k})<f(S^*)$. The chain $A_k \subseteq S_{A_k} \subseteq H'\cup S_{A_k} \subseteq S^*$ of inclusions thus lets us conclude $f(S_{A_k}) \leq f(H' \cup S_{A_k})$. Combined with~\eqref{eq:HAsIntSAk}, this implies
\begin{equation*}
f(H') \geq f(H)\enspace,
\end{equation*}
and the result now follows by the induction hypothesis, which implies $f(S^*)>f(H')$.
\end{proof}

\section{Disproving the existence of $(m,m-1)$-systems}
\label{sec:noBadSetSystem}

In this section we prove Theorem~\ref{thm:noMM-1Sys}, i.e., that no $(m,m-1)$-system exists for $m$ being a prime power. To this end, we present a variety of techniques to transform set systems into more structured ones. Using those transformations, we show that any $(m,m-1)$-system, for $m=p^{\alpha}$ being a prime power, could be transformed into a $(p,1)$-system $\mathcal{H}$, on a possibly different ground set, such that each set $H\in \mathcal{H}$ is in the same congruence class with respect to $\bmod\ p$, i.e., there is an $r\in \{0,\ldots, p-1\}$ with $|H| \equiv r \pmod{p}$ for all $H\in \mathcal{H}$.
Such systems can quite easily be seen not to exist, which is shown by the next lemma. 
Notice that the lemma does not depend on $p$ being a prime or prime power; only the transformations we introduce later depend on this.

\begin{lemma}\label{lem:structuredM1SystemNotPossible}
Let $N$ be a finite set. There is no non-empty intersection-closed set system $\mathcal{H}\subseteq 2^N$, and integers $p\in \mathbb{Z}_{>0}, r\in \{0,\ldots, p-1\}$ such that
\begin{enumerate}[label=(\roman*),itemsep=-0.2em,topsep=-0.2em]
\item\label{item:impSystemCardN} $|N|\not\equiv r \pmod{p}$, 
\item\label{item:impSystemCardH} $|H|\equiv r \pmod{p} \;\;\forall H\in \mathcal{H}$, and
\item\label{item:impSystemCovering} for any $e\in N$, there is a set $H\in\mathcal{H}$ with $e\in H$, i.e., $N=\bigcup_{H \in \mathcal{H}}H$.
\end{enumerate}
\end{lemma}

\begin{proof}
Assume for the sake of contradiction that there exists a set system $\mathcal{H}\subseteq 2^N$ with the properties stated in the lemma. We first observe that we can assume without loss of generality that $r=0$. Indeed, by introducing $p-r$ new elements that we add to $N$ and every set in $\mathcal{H}$, a new set system is obtained that fulfills the properties of the lemma with $r=0$. As $N=\bigcup_{H\in \mathcal{H}}H$, we can compute $|N|$ by the inclusion-exclusion principle:
\begin{align*}
|N| = \sum_{k=1}^{|\mathcal{H}|} (-1)^{k+1}
\sum_{\substack{\mathcal{F}\subseteq \mathcal{H}\\ |\mathcal{F}|=k}}
\left\vert \bigcap_{F\in \mathcal{F}} F\right\vert\enspace.
\end{align*}
However, when considering the above equation modulo $p$, a contradiction arises because each set $\bigcap_{F\in \mathcal{F}} F$ on the right-hand side is contained in $\mathcal{H}$, as $\mathcal{H}$ is intersection-closed, and thus $|\bigcap_{F\in \mathcal{F}} F|\equiv 0 \pmod{p}$; this implies that the right-hand side is $0 \bmod{p}$, which contradicts $|N|\not\equiv 0\pmod{p}$.
\end{proof}

To illustrate some of our techniques, we first present a transformation of sets systems that proves Theorem~\ref{thm:noMM-1Sys} for $m$ being a prime. Later, in Section~\ref{subsec:setTransformations}, we introduce a general framework of set transformations, which we can use, as we will show in Section~\ref{subsec:proofOfThmNoMM-1Sys}, to handle prime powers. Moreover, the versatility of these set transformations also allows us to extend our results to~\eqref{eq:GCCSM}, which we show in Section~\ref{sec:GCCSM}.

\subsection{Set transformations and nonexistence of $(m,m-1)$-systems for $m$ prime}
\label{subsec:mPrime}

To prove that no $(m,m-1)$-system exists for $m$ prime, assume for the sake of contradiction that there is an $(m,m-1)$-system $\mathcal{H}\subseteq 2^N$. Notice that without loss of generality we can assume that $|N|\equiv 0 \pmod{m}$, and consequently $|H|\not\equiv 0 \pmod{m}$ for $H\in \mathcal{H}$.
Indeed, if $|N|\equiv r \pmod{m}$, then we can construct a new set system by introducing $m-r$ new elements which get added to the ground set $N$ and also to every set in $\mathcal{H}$. One can easily observe that this leads to another $(m,m-1)$-system with $|N|\equiv 0 \pmod{m}$.

Our goal is now to transform $\mathcal{H}$ into a new set system, on a different ground set $W$, such that the cardinality of each set changes in a well-defined way. More precisely, we want that a set of cardinality $x$ gets transformed into a set of cardinality $g(x)=x^{m-1}$. For $m$ being prime, Fermat's Little Theorem implies $x^{m-1} \equiv 1 \pmod{m}$ for any $x\not\equiv 0 \pmod{m}$. Hence, such a transformation would have the desired effect that any set $H\in \mathcal{H}$ will be transformed to a set in the same congruence class; moreover, the cardinality of the image of the ground set would remain $0 \pmod{m}$. Furthermore, for the resulting system to be an $(m,1)$-system, we need two additional properties: First, each element of the new ground set needs to be contained in at least one transformed set, and additionally, the transformed system needs to retain the property of being intersection-closed.

We now describe how a set transformation $G\colon2^N \rightarrow 2^W$ with the properties described above can be obtained. The new ground set is
\begin{equation*}
W = N^{m-1} \coloneqq \underbrace{N\times N \times \ldots \times N}_{\text{$m-1$ times}}\enspace.
\end{equation*}
Moreover, a set $S\subseteq N$ gets transformed into the set
\begin{equation*}
G(S) = \{(e_1,\ldots, e_{m-1}) \mid e_1,\ldots, e_{m-1}\in S\}\subseteq W\enspace.
\end{equation*}
Clearly, the cardinality of the transformed set $G(S)$ is $|G(S)| = |S|^{m-1}$. Hence, the change of cardinalities is indeed described by the function $g(x) = x^{m-1}$, as desired. Hence, if we look at the transformed set system $G(\mathcal{H})\coloneqq\{G(H)\mid H\in\mathcal{H}\}\subseteq 2^W$ on the new ground set $W$, we have
\begin{enumerate}[label=(\roman*),leftmargin=1.5em,itemsep=-0.0em,topsep=0.5em]
\item  $|W| = g(|N|) = |N|^{m-1} \equiv 0 \pmod{m}$, because $|N|\equiv 0 \pmod{m}$, and
\item $|G(H)|\! = \! g(|H|)\! =\! |H|^{m-1}\! \equiv\! 1 \pmod{m} \;\forall H\in \mathcal{H}$, by Fermat's Little Theorem and $|H|\not\equiv 0 \pmod{m}$.
\end{enumerate}
Moreover, $G(\mathcal{H})$ is indeed intersection-closed because the definition of $G$ implies 
\begin{equation*}
G(S\cap T) = G(S) \cap G(T) \quad \forall S,T\subseteq N\enspace.
\end{equation*}
Finally, $G(\mathcal{H})$ is an $(m,1)$-system, as each element $(e_1, \ldots, e_{m-1})\in W$ is covered by a set in $G(\mathcal{H})$ due to the following. Because $\mathcal{H}$ is an $(m,m-1)$-system, there is a set $H\in \mathcal{H}$ such that $\{e_1,\ldots, e_{m-1}\}\subseteq H$, and hence $(e_1, \ldots, e_{m-1})\in G(H)$.
Hence, $G(\mathcal{H})$ is an $(m,1)$-system with all sets in $G(\mathcal{H})$ being in the same congruence class $\bmod\ m$, which, by Lemma~\ref{lem:structuredM1SystemNotPossible}, does not exist and thus leads to the desired contradiction.
This disproves the existence of $(m,m-1)$-systems for $m$ being a prime, and implies via Theorem~\ref{thm:EnumDGoodIfNoBadSys} that our enumeration procedure works for prime moduli.

\begin{corollary}
For $m$ being a prime, \refEnumD with $d=m-1$ returns an optimal solution to \eqref{eq:CCSM} with modulus $m$.
\end{corollary}

Whereas the above product space transformation was enough to deal with prime moduli and allowed for highlighting several important ideas, we need more involved transformations to deal with prime powers and~\eqref{eq:GCCSM}. In the next section, we therefore formalize and discuss in more generality a large class of cardinality transformations $g$ that can be achieved, and how they can be combined.

\subsection{A general framework based on set transformations}
\label{subsec:setTransformations}

We start by formalizing the idea of a transformation that changes the cardinality of a set $S$ in a well-defined way by some function $g$ and will also maintain the intersection-closed property, analogous to the transformation described in Section~\ref{subsec:mPrime}.
Moreover, the notion of the \emph{level of $g$}, which we also define below, allows us to give a simple condition to guarantee that elements in the new ground set remain covered by transformed sets.

\begin{definition}\label{def:cardTrans}
A map $g\colon\mathbb{Z}_{\geq 0} \rightarrow \mathbb{Z}_{\geq 0}$ is a \emph{cardinality transformation function} if for every finite set $N$, there is a finite set $W$ and a map $G\colon2^N \rightarrow 2^W$ such that
\begin{enumerate}[itemsep=-0.2em,topsep=0.2em,label=(\roman*)]
\item $G(N) = W$,
\item $|G(S)| = g(|S|) \;\;\forall S\subseteq N$, and
\item\label{item:ctfIntClosed} $G(S)\cap G(T) = G(S\cap T) \;\;\forall S,T\subseteq N$.
\end{enumerate}
Moreover, for $\ell\in \mathbb{Z}_{\geq 0}$, we say that $g$ is of \emph{level $\ell$} if $G$ can be chosen such that for every $w\in W$, there exists a set $S\subseteq N$ with $|S|\leq \ell$ such that $w\in G(S)$. In this case we call $G$ a set transformation of level $\ell$.

We call $G$ a \emph{$g$-realizing set transformation} for the ground set $N$. Conversely, $g$ is called the \emph{cardinality transformation function corresponding to $G$}. 
\end{definition}

Notice that property~\ref{item:ctfIntClosed} implies that for any intersection-closed family $\mathcal{H}\subseteq 2^N$, its image $G(\mathcal{H})$ is as well intersection-closed.
Furthermore, a set transformation function is always monotone, i.e., $G(S) \subseteq G(T)$ for $S\subseteq T \subseteq N$. This follows from $G(S) = G(S\cap T) = G(S) \cap G(T) \subseteq G(T)$ for any $S\subseteq T$.
Moreover, we recall that we want to find a set transformation $G$ that would transform an $(m,m-1)$ system, for $m$ being a prime power, to a system with the properties stated in Lemma~\ref{lem:structuredM1SystemNotPossible}, which leads to a contradiction by the same lemma. Hence, we want to find a set transformation $G$ such that the transformed set system still covers the ground set.
For this, observe that by applying a set transformation of level $\ell$ to any set system $\mathcal{H}\subseteq 2^N$ satisfying that for any $U\subseteq N$ with $|U|\leq \ell$, there is a set $S\in \mathcal{H}$ such that $U\subseteq S$, a new set system that covers the whole ground set is obtained. To better quantify this property in a way that allows us later to combine several set transformations, we introduce the notion of a $k$-covering set system.

\begin{definition}
For $k\in \mathbb{Z}_{\geq 1}$, a set family $\mathcal{H}\subseteq 2^N$ is \emph{$k$-covering} if, for any $U\subseteq N$ with $|U|\leq k$, there exists a set $S\in \mathcal{H}$ such that $U\subseteq S$.
\end{definition}

Hence, any $(m,d)$-system (and also any $(\mathcal{F},d)$-system) is a $d$-covering set system by definition. Moreover, a $1$-covering set system is a system covering the whole ground set. 
The following observation highlights how the coverage of a set system changes through set transformations of a certain level.

\begin{lemma}\label{lem:transCoverage}
Let $\mathcal{H}\subseteq 2^N$ be a $k$-covering set system, and let $G$ be a set transformation of level $\ell\in \mathbb{Z}_{\geq 1}$. Then $G(\mathcal{H})$ is a $\left\lfloor \frac{k}{\ell}\right\rfloor$-covering system.
\end{lemma}
\begin{proof}
Let $W=G(N)$ be the ground set of the transformed set system $G(\mathcal{H})$, and let $U\subseteq W$ with $|U|\leq \lfloor \frac{k}{\ell}\rfloor$. We have to show that there is a set $Y\in G(\mathcal{H})$ with $U\subseteq Y$. Because $G$ is of level $\ell$, for each element $u\in U$ there is a set $S_u\subseteq N$ with $|S_u|\leq \ell$ and $u\in G(S_u)$. Notice that $S\coloneqq \bigcup_{u\in U} S_u$ has thus size at most $|S| \leq \ell |U| \leq k$. Because $\mathcal{H}$ is $k$-covering, there exists $X\in \mathcal{H}$ with $S\subseteq X$. We finish the proof by showing that for $Y=G(X)$ we indeed have $U \subseteq Y$, which holds because we have that for all $u\in U$,
\begin{equation*}
G(X) \supseteq G(S) \supseteq G(S_u) \ni u\enspace,
\end{equation*}
where we use monotonicity of $G$ on the sets $X\supseteq S \supseteq S_u$, and the fact that $u\in G(S_u)$.
\end{proof}

In summary, the following provides a sufficient condition to disprove the existence of an $(m,m-1)$-system. Note that in the following statement, the number $p$ is not required to be prime.

\begin{theorem}\label{thm:targetCTF}
Let $m\in \mathbb{Z}_{\geq 1}$ and $d\in \mathbb{Z}_{\geq 1}$. There does not exist an $(m,d)$-system if there exists an integer $p\in \mathbb{Z}_{\geq 1}$ and a cardinality transformation function $g$ of level $d$ such that for $x\in \mathbb{Z}_{\geq 0}$:
\begin{equation}\label{eq:goodCTF}
g(x) \equiv \begin{cases}
0 \pmod{p} & \text{if } x\equiv 0 \pmod{m},\\
1 \pmod{p} & \text{if } x\not\equiv 0 \pmod{m}.
\end{cases}
\end{equation}
\end{theorem}
\begin{proof}
With the goal of deriving a contradiction, assume that there exists both a cardinality transformation function as stated in the theorem and an $(m,d)$-system $\mathcal{H}\subseteq 2^N$ on some finite ground set $N$.
Let $r\in \{0,\ldots, m-1\}$ be such that $|N| \equiv r \pmod{m}$. Notice that, as in the proof of the above statement for prime numbers $m$ in Section~\ref{subsec:mPrime}, we can assume $r=0$, because if $r\neq 0$, then we can add $m-r$ new elements to $N$ and each set in $\mathcal{H}$, thus obtaining an $(m,d)$-system on a larger ground set with $r=0$. Hence, assume $r=0$.
The theorem now follows by observing that $G(\mathcal{H})$, where $G$ is a $g$-realizing set transformation for $N$ of level $d$, is a set system system fulfilling the conditions of Lemma~\ref{lem:structuredM1SystemNotPossible} with $r=1$, which is impossible by the same lemma. Notice that the fact of $G(\mathcal{H})$ covering the whole transformed ground set $G(N)$ follows by Lemma~\ref{lem:transCoverage}, as any $(m,d)$-system is by definition $d$-covering, and $G$ is of level $d$.
\end{proof}

The following two lemmas present a large class of cardinality transformation functions with low level. In Section~\ref{subsec:proofOfThmNoMM-1Sys}, we will see that this class is rich enough to disprove the existence of $(m,m-1)$-systems for $m$ being a prime power via Theorem~\ref{thm:targetCTF}.

\begin{lemma}\label{lem:basicCTF}
The following cardinality transformation functions $g\colon\mathbb{Z}_{\geq 0} \rightarrow \mathbb{Z}_{\geq 0}$ exist for every $k\in \mathbb{Z}_{\geq 1}$:
\begin{enumerate}[label=(\roman*),itemsep=-0.2em,topsep=0.2em]
\item\label{item:basicCTFConst} $g(x) = k$ of level $0$,
\item\label{item:basicCTFMonomial} $g(x) = x^{k}$ of level $k$, and
\item\label{item:basicCTFBinomial} $g(x) = \binom{x}{k}$ of level $k$.\footnote{We employ the usual convention that $\binom{n}{k}=0$ for $k>n$.}
\end{enumerate}
\end{lemma}
\begin{proof}
Throughout this proof, let $N$ be an arbitrary finite ground set. We have to show that there is a $g$-realizing cardinality transformation function $G$ for $N$ of the claimed level.

\begin{enumerate}[label=(\roman*),itemsep=-0.2em,topsep=0.2em]
\item Let $W$ be a set of cardinality $k$, and we define $G(S) = W$ for every $S\subseteq N$. One can easily verify that $G$ is a $g$-realizing cardinality transformation function for $g(x) =k$; moreover, it has level $0$ since $G(\emptyset) = W$.

\item The existence of such a cardinality transformation function was shown in our example in Section~\ref{subsec:mPrime}, where $k$ corresponds to $m-1$.

\item For a finite set $A$ and $a \geq 1$, we denote by $\binom{A}{a}$ the family of all subsets of $A$ of cardinality $a$. 
The transformed ground set $W=G(N)$ is set to be $W = \binom{N}{k}$, and we define $G\colon2^N \rightarrow 2^W$ to be $G(S) = \binom{S}{k}$, i.e., this is the family of all subsets of size $k$ of elements in $S$, also called $k$-subsets of $S$. This $G$ clearly fulfills $G(N)=W$ and $|G(S)| = g(S)$ for all $S\subseteq N$. Moreover, for $S,T\subseteq N$, we have
\begin{equation*}
G(S) \cap G(T) = \{\text{all $k$-subsets of $S$}\} \cap
\{\text{all $k$-subsets of $T$}\}
 = \{\text{all $k$-subsets of $S\cap T$}\}
 = G(S\cap T)\enspace.
\end{equation*}
Finally, the level of $G$ is indeed $k$, because any $w\in W$ corresponds to a $k$-subset of $N$, i.e., $w=\{e_1,\ldots, e_k\}\subseteq N$, and we have $G(\{e_1,\ldots, e_k\}) = \{w\}$.\qedhere
\end{enumerate}
\end{proof}

\begin{lemma}\label{lem:addCTF}
Let $g_1, g_2 \colon \mathbb{Z}_{\geq 0} \rightarrow \mathbb{Z}_{\geq 0}$ be two cardinality transformation functions of level $\ell_1$ and $\ell_2$, respectively. Then $g_1 + g_2$ is a cardinality transformation function of level $\max\{\ell_1,\ell_2\}$.
\end{lemma}
\begin{proof}
Let $N$ be a finite ground set, and let $G_i\colon2^N \rightarrow 2^{W_i}$ for $i\in \{1,2\}$ be a $g_i$-realizing set transformation of level $\ell_i$. Moreover, we choose the sets $W_1=G_1(N)$ and $W_2=G_2(N)$ to be disjoint. We claim that $G\colon2^N \rightarrow 2^{W_1\cup W_2}$ defined by $G(S) = G_1(S) \cup G_2(S)$ is a $(g_1+g_2)$-realizing set transformation of level $\max\{\ell_1,\ell_2\}$ as desired. Indeed, $G(N) = G_1(N) \cup G_2(N) = W_1\cup W_2$. Moreover, for any $S\subseteq N$,
\begin{align*}
|G(S)|&=|G_1(S) \cup G_2(S)| = |G_1(S)| + |G_2(S)| 
= g_1(|S|) + g_2(|S|) = (g_1+g_2)(|S|)\enspace,
\end{align*}
where the second equality follows from $G_i(S) \subseteq W_i$ for $i\in \{1,2\}$ and $W_1$ and $W_2$ were chosen to be disjoint.
Furthermore, for any $S,T\subseteq N$, we have
\begin{align*}
G(S) \cap G(T) &= (G_1(S) \cup G_2(S)) \cap (G_1(T) \cup G_2(T)) = (G_1(S) \cap G_1(T)) \cup (G_2(S) \cap G_2(T))\\
&=  G_1(S\cap T) \cup G_2(S\cap T) = G(S\cap T)\enspace,
\end{align*}
again exploiting disjointness of images with respect to $G_1$ and $G_2$, and the fact that $G_1$ and $G_2$ fulfill the intersection property of set transformations.
Finally, $G$ is indeed of level $\max\{\ell_1,\ell_2\}$, because for any element $w\in W_1\cup W_2$ there is an $i\in \{1,2\}$ such that $w\in W_i$, and due to the fact that $G_i$ is of level $\ell_i$, there exists a set $S\subseteq N$ with $|S|\leq \ell_i$ and $w\in G_i(S) \subseteq G(S)$.
\end{proof}

By combining Lemma~\ref{lem:basicCTF} and Lemma~\ref{lem:addCTF}, we obtain the following.

\begin{corollary}\label{cor:polyBinomCTF}
For any $k\in \mathbb{Z}_{\geq 1}$ and $a$, $b_1, \ldots, b_k$, $c_1,\ldots, c_k \in \mathbb{Z}_{\geq 0}$, the function
\begin{equation*}
g(x) = a + \sum_{i=1}^k b_i x^i + \sum_{i=1}^k c_i \binom{x}{i}
\end{equation*}
is a cardinality transformation function of level $k$.
\end{corollary}

Theorem~\ref{thm:targetCTF} together with the existence of a rich set of cardinality transformation functions of low level, as stated by the above corollary, lead to a general approach to disprove the existence of $(m,d)$-systems in a concise way. Moreover, the approach can be adjusted to further settings, as we will see in Section~\ref{sec:GCCSM}, when talking about~\eqref{eq:GCCSM}.

In particular, the proof of why no $(m,m-1)$-system exists for $m$ being prime can now be rephrased as follows in a concise way. By Corollary~\ref{cor:polyBinomCTF} (or even just by Lemma~\ref{lem:basicCTF}) the function $g(x)=x^{m-1}$ is a cardinality transformation function of level $m-1$; moreover, it has property~\eqref{eq:goodCTF} stated in Theorem~\ref{thm:targetCTF} for $p=m$, due to Fermat's Little Theorem. Thus, by Theorem~\ref{thm:targetCTF}, an $(m,m-1)$-system, for $m$ prime, does not exist, which, by Theorem~\ref{thm:EnumDGoodIfNoBadSys}, implies that \refEnumD[m-1] returns an optimal solution to any~\eqref{eq:CCSM} problem with modulus $m$, as desired.

\subsection{Proof of Theorem~\ref{thm:noMM-1Sys}}
\label{subsec:proofOfThmNoMM-1Sys}

To disprove the existence of an $(m,m-1)$-system for $m=p^{\alpha}$ being a prime power, we consider the following cardinality transformation function of level $m-1$, whose existence is guaranteed by Corollary~\ref{cor:polyBinomCTF}:
\begin{equation}\label{eq:gForPrimePowers}
g(x) = \sum_{\substack{1\leq k < m,\\ k \text{ odd}}}
\binom{x}{k}
+ (p-1)\sum_{\substack{1\leq k < m,\\ k \text{ even}}}
\binom{x}{k}\enspace.
\end{equation}

To show in Lemma~\ref{lem:gForPrimePowers} that $g$ fulfils the conditions of Theorem~\ref{thm:targetCTF}, we use the following relation for binomial coefficients over a field $\mathbb{F}_p$ for $p$ prime, which follows from elementary techniques.%

\begin{lemma}\label{lem:binomMod}
Let $p$ be a prime, and let $a,b\in \mathbb{Z}_{\geq 0}$ and $\alpha \in \mathbb{Z}_{\geq 1}$ with $b< p^{\alpha}$. Then, it holds that
\begin{equation*}
\binom{a}{b} \equiv \binom{a \bmod p^{\alpha}}{b} \pmod{p}\enspace.
\end{equation*}
\end{lemma}

\begin{proof}
As a first step, we show that if $a - p^\alpha\geq 0$, then $\binom{a}{b}\equiv \binom{a-p^\alpha}{b}\pmod{p}$. Using Vandermonde's identity, we obtain
\begin{equation*}
\binom{a}{b} = \binom{(a-p^\alpha)+p^\alpha}{b} = \sum_{k=0}^b \binom{a-p^\alpha}{b-k}\binom{p^\alpha}{k}\enspace.
\end{equation*}
Note that $p^\alpha>k\geq 1$ implies that $\binom{p^\alpha}{k} = \frac{p^\alpha}{k}\binom{p^\alpha-1}{k-1}$ is divisible by $p$. As $b<p^\alpha$, we see that after reducing the above equation $\bmod\ p$, the only possibly non-zero summand is $\binom{a-p^\alpha}{b}$, as desired.
Iteratively applying $\binom{a}{b}\equiv \binom{a-p^\alpha}{b}\pmod{p}$, we immediately obtain that for all $\ell\in\mathbb{Z}_{\geq0}$ satisfying $a-\ell p^\alpha \geq 0$, we have
\begin{equation*}
\binom{a}{b} \equiv \binom{a-\ell p^\alpha}{b} \pmod{p}\enspace.
\end{equation*}
Choosing $\ell = \left\lfloor\frac{a}{p^\alpha}\right\rfloor$, we get that $a-\ell p^\alpha\in\{0,\ldots,p^\alpha-1\}$ is the residue class of $a$ $\bmod\ p^\alpha$, and the result follows.
\end{proof}

\begin{lemma}\label{lem:gForPrimePowers}
Let $m=p^{\alpha}$ be a prime power. Then, the function $g$ defined by~\eqref{eq:gForPrimePowers} fulfills property~\eqref{eq:goodCTF}, i.e., for $x\in \mathbb{Z}$ we have
\begin{equation*}
g(x) \equiv \begin{cases}
0 \pmod{p} & \text{if } x\equiv 0 \pmod{m},\\
1 \pmod{p} & \text{if } x\not\equiv 0 \pmod{m}.
\end{cases}
\end{equation*}
\end{lemma}

\begin{proof}
Let $x\in \mathbb{Z}$. Due to Lemma~\ref{lem:binomMod}, we have $g(x) = g(x\bmod p)$. Hence, we can assume $x\in \{0,\ldots, m-1\}$.
For $x=0$ we clearly have $g(x) = 0$. Thus, assume $x\in \{1,\ldots, m-1\}$, and it remains to show $g(x)\equiv 1 \pmod{p}$, which holds due to
\begin{align*}
g(x) &= \sum_{\substack{1\leq k < m,\\ k \text{ odd}}}
\binom{x}{k}
+ (p-1)\sum_{\substack{1\leq k < m,\\ k \text{ even}}}
\binom{x}{k}\\
  &\equiv \sum_{\substack{1\leq k < m,\\ k \text{ odd}}}
\binom{x}{k}
- \sum_{\substack{1\leq k < m,\\ k \text{ even}}}
\binom{x}{k}&\pmod{p} \\
  &= 1-\sum_{k=0}^{m-1} (-1)^{k} \binom{x}{k} \\
  &= 1-\sum_{k=0}^{x} (-1)^{k} \binom{x}{k} \\
  &= 1-(1-1)^x = 1\enspace.  \tag*{\qedhere}
\end{align*}
\end{proof}

Combining the results of Lemma~\ref{lem:gForPrimePowers} and Theorem~\ref{thm:targetCTF}, we obtain that for any prime power $m=p^\alpha$, there does not exist an $(m,m-1)$-system, which finishes the proof of Theorem~\ref{thm:noMM-1Sys}.

\section{Extension to~\eqref{eq:GCCSM}}
\label{sec:GCCSM}

The methods for proving Theorem~\ref{thm:mainGCCSM}, which states polynomial time solvability of \eqref{eq:GCCSM} for prime power moduli, closely follow those presented above for \eqref{eq:CCSM}. As before, we establish a link between failure of the algorithm \refEnumD and set systems with certain properties. While this link leads to $(m,d)$-systems for \eqref{eq:CCSM} problems, we need the more general notion of $(m,k,d)$-systems for \eqref{eq:GCCSM}. For two vectors $x,y\in\mathbb{Z}^k$, we write $x\not\equiv y\pmod*{m}$ if there exists $i\in [k]$ with $x_i\not\equiv y_i\pmod*{m}$. 

\begin{definition}[$(m,k,d)$-system (with respect to $(S_1,\ldots,S_k)$ on $N$)]\label{def:mkdSystem}
Let $N$ be a finite ground set, let $m,k,d\in \mathbb{Z}_{>0}$, and let $S_1,\ldots,S_k\subseteq N$. We say that a set system $\mathcal{H}\subseteq 2^N$ is an $(m,k,d)$-system (with respect to $(S_1,\ldots,S_k)$ on $N$) if
\begin{enumerate}[label=(\roman*),itemsep=-0.2em,topsep=-0.2em]
\item\label{item:MKDIntClosed} $\mathcal{H}$ is closed under intersections,

\item\label{item:MKDDiffParity} $(|H\cap S_1|,\ldots,|H\cap S_k|) \not\equiv (|S_1|,\ldots,|S_k|) \pmod{m} \quad\forall H\in \mathcal{H}$, and

\item\label{item:MKDCoverage} for any $S\subseteq N$ with $|S|\leq d$, there is a $H\in \mathcal{H}$ with $S\subseteq H$.
\end{enumerate}
\end{definition}

Note that property~\ref{item:MKDCoverage} precisely states that every $(m,k,d)$-system is $d$-covering. Using the tools developed in Section~\ref{sec:reduceToSetSys}, we can immediately prove the following analogon to Theorem~\ref{thm:EnumDGoodIfNoBadSys}, thus reducing correctness of \refEnumD to a combinatorial question about nonexistence of $(m,k,d)$-systems. As for $(m,d)$-systems, nonexistence of $(m,k,d)$-systems without explicit reference to a ground set is to be understood to hold for any ground set, i.e., for every finite ground set $N$, there does not exist an $(m,k,d)$-system on $N$.

\begin{theorem}
Let $m,k,d\in\mathbb{Z}_{>0}$. If no $(m,k,d)$-system exists, then \refEnumD returns an optimal solution to any \eqref{eq:GCCSM} problem with modulus $m$ and $k$ congruency constraints.
\end{theorem}

\begin{proof}
Consider a \eqref{eq:GCCSM} problem $\min\{f(S) \mid S\in \mathcal{L}, \; |S\cap S_i| \equiv r_i \pmod*{m} \;\;\forall i\in [k]\}$ with $k$ congruency constraints. Consequently, the set family $\mathcal{F}$ over which we want to minimize the function $f$ is given by
\begin{equation*}
\mathcal{F} = \{ S\in \mathcal{L} \mid |S\cap S_i| \equiv r_i \pmod*{m} \;\;\forall i\in [k]\}\enspace.
\end{equation*}
By Theorem~\ref{thm:EnumDGoodIfNoBadSysGen}, it is sufficient to see that no $(\mathcal{F},d)$-systems and no $(\comp(\mathcal{F}),d)$-systems exist. To finish the proof, we show that each of these two systems are also $(m,k,d)$-systems. To this end, consider an $(\mathcal{F}, d)$-system $\mathcal{H}$, and let $Q=\bigcup_{H\in\mathcal{H}} H$. Without loss of generality, we may assume that $S_i\subseteq Q$ for all $i\in [k]$ (if not, we simply delete the elements in $S_i\setminus Q$ for all $i\in [k]$). By property~\ref{item:FDQin} of an $(\mathcal{F},d)$-system, we have $Q\in\mathcal{F}$, and hence $|S_i| = |Q\cap S_i| \equiv r_i \pmod{m} \;\forall i\in [k]$. Together with property~\ref{item:FDSetsNotFeasible} of $(\mathcal{F},d)$-systems, this implies property~\ref{item:MKDDiffParity} of an $(m,k,d)$-system. Moreover, properties~\ref{item:FDIntClosed} and~\ref{item:FDCoverage} of an $(\mathcal{F},d)$-system correspond to properties~\ref{item:MKDIntClosed} and~\ref{item:MKDCoverage} of an $(m,k,d)$-system.

Moreover, for a $(\comp(\mathcal{F}),d)$-system, note that we have
\begin{align*}
\comp(\mathcal{F}) &= \{ S\in\comp(\mathcal{L}) \mid |(N\setminus S)\cap S_i| \equiv r_i \pmod*{m} \;\;\forall i\in [k]\}\\
&= \{ S\in\comp(\mathcal{L}) \mid |S\cap S_i| \equiv |N\cap S_i|-r_i \pmod*{m} \;\;\forall i\in [k]\}\enspace,
\end{align*}
so $\comp(\mathcal{F})$ has the same form as $\mathcal{F}$, and is therefore also an $(m,k,d)$-system.
\end{proof}

The previous theorem implies that in order to prove Theorem~\ref{thm:mainGCCSM}, it remains to show that no $(m,k,d)$-systems exist for $m$ being a prime power and some $d = km + O(1)$. This is the content of the following theorem.

\begin{theorem}\label{thm:noMKKM-1Sys}
For $m\in\mathbb{Z}_{>0}$ being a prime power, there is no $(m,k,k(m-1))$-system.
\end{theorem}

The idea for proving Theorem~\ref{thm:noMKKM-1Sys} is to assume existence of an $(m,k,k(m-1))$-system and apply set system transformations to obtain more structured systems. More precisely, our proof involves two transformations. First, we apply a transformation very similar to the one given in~\eqref{eq:gForPrimePowers} that we used in the proof of Theorem~\ref{thm:noMM-1Sys}. Through this transformation, we obtain a well-structured $(m,k,k)$-system in which the vectors $(|H\cap S_1|,\ldots,|H\cap S_k|)$ take only very restricted values $\bmod\ p$, where $p$ is the prime such that $m=p^{\alpha}$ for some $\alpha\in \mathbb{Z}_{\geq 1}$.
In a second step, we show that the previously obtained system can in turn be transformed to a system contradicting Lemma~\ref{lem:structuredM1SystemNotPossible}. This step requires a more general type of transformation functions than the ones seen before, which we introduce in the next section, before finally proving Theorem~\ref{thm:noMKKM-1Sys}.

\subsection{Set transformations for the generalized setting}

Generalized cardinality transformation functions are very similar to the cardinality transformation functions seen earlier in Definition~\ref{def:cardTrans}. Here, the cardinality $|G(S)|$ of a transformed set $S\subseteq N$ depends on the sizes of $|S\cap S_i|$ for $i\in [k]$, instead of just the size of $S$. Formally, the definition is as follows.

\begin{definition}\label{def:genCardTrans}
A map $g\colon\mathbb{Z}^k_{\geq 0} \rightarrow \mathbb{Z}_{\geq 0}$ is a \emph{generalized cardinality transformation function} if for every finite set $N$ and all sets $S_1,\ldots,S_k\subseteq N$, there is a finite set $W$ and a map $G\colon2^N \rightarrow 2^W$ such that
\begin{enumerate}[itemsep=-0.2em,topsep=0.2em,label=(\roman*)]
\item $G(N) = W$,
\item\label{item:gctfCardinalities} $|G(S)| = g(|S\cap S_1|,\ldots,|S\cap S_k|) \;\;\forall S\subseteq N$, and
\item\label{item:gctfIntClosed} $G(S)\cap G(T) = G(S\cap T) \;\;\forall S,T\subseteq N$.
\end{enumerate}
Moreover, for $\ell\in \mathbb{Z}_{\geq 1}$, we say that $g$ is of \emph{level $\ell$} if $G$ can be chosen such that for every $w\in W$, there exists a set $S\subseteq N$ with $|S|\leq \ell$ such that $w\in G(S)$. In this case we call $G$ a set transformation of level $\ell$.

We call $G$ a \emph{$g$-realizing set transformation} for the ground set $N$ and the sets $S_1,\ldots,S_k$. Conversely, $g$ is called the \emph{cardinality transformation function corresponding to $G$}. 
\end{definition}

As pointed out before, the only difference to cardinality transformation functions as introduced in Definition~\ref{def:cardTrans} is property~\ref{item:gctfCardinalities}. For this reason, the properties that we proved for cardinality transformation functions also hold true for generalized cardinality transformation functions. In particular, if $G$ is a set transformation function of level $\ell$ realizing a generalized cardinality transformation function, and $\mathcal{F}$ is a set system, we have the following. If $\mathcal{F}$ is intersection-closed, then so is $G(\mathcal{F})$ (this follows from property~\ref{item:gctfIntClosed} above), and if $\mathcal{F}$ is $k$-covering, then $G(\mathcal{F})$ is $\lfloor\frac{k}{\ell}\rfloor$-covering (analogous to Lemma~\ref{lem:transCoverage}). Moreover, $G$ is a monotone function.

There are various ways to construct generalized cardinality transformation functions, but we restrict our attention to the precise function that we need for our proofs.

\begin{lemma}\label{lem:productGCTF}
For every $k\in\mathbb{Z}_{\geq 1}$, the function $g(x_1,\ldots,x_k) = x_1x_2\cdots x_k$ is a generalized cardinality transformation function of level $k$.
\end{lemma}

\begin{proof}
Let $N$ be a finite set and let $S_1,\ldots,S_k\subseteq N$. Let $W=S_1\times\ldots\times S_k$ and define $G\colon 2^N\rightarrow 2^W$ by
\begin{equation*}
G(S) = (S\cap S_1) \times \ldots \times (S\cap S_k)
\end{equation*}
for all $S\subseteq N$. We claim that $G$ is a $g$-realizing set transformation function. Indeed, it is easy to see that $G(N)=W$ by definition. Moreover, we have
\begin{align*}
|G(S)|=|(S\cap S_1) \times \ldots \times (S\cap S_k)| = |S\cap S_1| \cdot\ldots\cdot |S\cap S_k| = g(|S\cap S_1|, \ldots, |S\cap S_k|)\enspace.
\end{align*}
To see that $G$ also fulfills property~\ref{item:gctfIntClosed} in Definition~\ref{def:genCardTrans}, note that for all sets $S,T\subseteq N$, having $e\in(S\cap T\cap S_1) \times \ldots \times (S\cap T\cap S_k)$ is equivalent to having $e\in(S\cap S_1) \times \ldots \times (S\cap S_k)$ and $e\in(T\cap S_1) \times \ldots \times (T\cap S_k)$. Hence, $G(S\cap T) = G(S)\cap G(T)$, as desired.

To see that $g$ is of level $k$, note that every $w\in W$ is a sequence of elements $(s_1,\ldots,s_k)$ with $s_i\in S_i$ for $i\in [k]$. Let $S_w = \{s_1,\ldots,s_k\}$, then $w\in G(S_w)$ and $|S_w| \leq k$ (notice that we may have $|S_w|< k$, because some of $s_i$ may be identical). Thus $g$ is of level $k$.
\end{proof}

\subsection{Disproving existence of $(m,k,k(m-1))$-systems}

As outlined above, the proof of Theorem~\ref{thm:noMKKM-1Sys}, namely that there do not exist $(m,k,k(m-1))$-systems, has two steps. In a first step, we disprove the existence of a very structured version of an $(m,k,k)$-system. In a second step, we prove Theorem~\ref{thm:noMKKM-1Sys} by showing that any $(m,k,k(m-1))$-system for $m$ being a prime power can be reduced to this structured version of an $(m,k,k)$-system.

\begin{lemma}\label{lem:GCTFapplication}
Let $m,k,p\in\mathbb{Z}_{\geq 0}$, let $N$ be a finite set and let $S_1,\ldots,S_k\subseteq N$. There does not exist a non-empty $(m,k,k)$-system $\mathcal{H}$ with respect to $(S_1,\ldots,S_k)$ on $N$ such that
\begin{enumerate}[label=(\roman*),itemsep=-0.2em,topsep=-0.2em]
\item $|S_i|\equiv 1 \pmod{p} \;\;\forall i\in [k]$, and
\item $(|H\cap S_1|,\ldots,|H\cap S_k|)\in\{0,1\}^k\setminus\{(1,\ldots,1)\} \pmod{p} \;\;\forall H\in\mathcal{H}$.
\end{enumerate}
\end{lemma}

\begin{proof}
Fix $m,k\in\mathbb{Z}_{\geq 0}$, a finite set $N$ and $S_1,\ldots,S_k\subseteq N$, and assume with the goal of deriving a contradiction that the system $\mathcal{H}$ specified in Lemma~\ref{lem:GCTFapplication} exists.
Consider the generalized cardinality transformation function $g(x_1,\ldots,x_k)=x_1\cdots x_k$, and let $G$ be a $g$-realizing set transformation of level $k$ for the ground set $N$ and the sets $S_1,\ldots,S_k$, whose existence is guaranteed by Lemma~\ref{lem:productGCTF}.

The lemma now follows by observing that $G(\mathcal{H})$ is a set system that satisfies all conditions of Lemma~\ref{lem:structuredM1SystemNotPossible} with $r=0$. Indeed, we see that the new ground set $G(N)$ has cardinality $|G(N)|=g(|S_1|,\ldots,|S_k|)=|S_1|\cdot\ldots\cdot |S_k|\equiv1\pmod{p}$ by the first assumption, settling property~\ref{item:impSystemCardN} in the assumptions of Lemma~\ref{lem:structuredM1SystemNotPossible}. On the other hand, every set in $G(\mathcal{H})$ is of the form $G(H)$ for some $H\in\mathcal{H}$, and has cardinality $|G(H)|=g(|H\cap S_1|,\ldots,|H\cap S_k|)=|H\cap S_1|\cdot\ldots\cdot|H\cap S_k|\equiv 0\pmod{p}$ because, by the second assumption, at least one of the factors vanishes $\bmod\ {p}$. This proves that $G(\mathcal{H})$ has property~\ref{item:impSystemCardH} of Lemma~\ref{lem:structuredM1SystemNotPossible}. Property~\ref{item:impSystemCovering} follows from the fact that $\mathcal{H}$ is $k$-covering and $g$ is of level $k$, hence $G(\mathcal{H})$ is still $1$-covering. Moreover, as the image of a non-empty intersection-closed set system, $G(\mathcal{H})$ is non-empty and intersection-closed, as well.
This shows that $G(\mathcal{H})$ fulfills all conditions of Lemma~\ref{lem:structuredM1SystemNotPossible}. Consequently, by the same lemma, we obtain the desired contradiction.
\end{proof}

\begin{proof}[Proof of Theorem~\ref{thm:noMKKM-1Sys}]
We assume for the sake of contradiction that for some prime power $m=p^\alpha$, there exists an $(m,k,k(m-1))$-system $\mathcal{H}$ with respect to $(S_1,\ldots,S_k)$ on $N$ for some finite ground set $N$ and subsets $S_i\subseteq N$ for $i\in [k]$.
If, for some $i\in [k]$, $|S_i|\not\equiv 0\pmod{m}$, let $r_i\in\{1,\ldots,m-1\}$ such that $|S_i|\equiv r_i \pmod{m}$. We introduce $m-r_i$ new elements and add them to $S_i$ and all sets in $\mathcal{H}$ to obtain a new $(m,k,k(m-1))$-system with $|S_i|\equiv 0\pmod{m}$. After doing so for all $i\in [k]$ with $|S_i|\not\equiv 0\pmod{m}$, we obtain a corresponding set system with $|S_i|\equiv 0\pmod{m}$ for all $i\in[k]$.

Let $g$ be the cardinality transformation function of level $m-1$ defined in \eqref{eq:gForPrimePowers}, and let $g'\colon \mathbb{Z}_{\geq0}\to\mathbb{Z}_{\geq0}$ be defined by $g'(x)=1+(p-1) g(x)$. By Lemma~\ref{lem:addCTF}, $g'$ is a cardinality transformation function of level $m-1$. Moreover, by Lemma~\ref{lem:gForPrimePowers}, we have
\begin{equation}\label{eq:g'mod}
g'(x) \equiv 1-g(x) \equiv \begin{cases}
1 \pmod{p} & \text{if } x\equiv 0 \pmod{m},\\
0 \pmod{p} & \text{if } x\not\equiv 0 \pmod{m}.
\end{cases}
\end{equation}
Let $G'$ be a $g'$-realizing set transformation function, and note that $G'(\mathcal{H})$ is an $(m,k,k)$-system with respect to $(G'(S_1),\ldots,G'(S_k))$ on $G'(N)$. 
To see this, we verify the properties in Definition~\ref{def:mkdSystem}. Note that $G'(\mathcal{H})$ is indeed closed under intersections because $\mathcal{H}$ is, and $G'$ preserves intersections. Furthermore, by \eqref{eq:g'mod}, we have 
\begin{equation}\label{eq:setsTo1}
|G'(S_i)|\equiv g'(|S_i|)\equiv 1\pmod{p}
\end{equation}
for all $i\in[k]$. Moreover, any set in $G'(\mathcal{H})$, which is of the form $G'(H)$ for some $H\in\mathcal{H}$, fulfills
\begin{equation}\label{eq:vecsNotTo1}
\begin{aligned}
(|G'(H)\cap G'(S_1)|,\ldots,|G'(H)\cap G'(S_k)|) &= (|G'(H\cap S_1)|,\ldots,|G'(H\cap S_k)|)\\
&= (g'(|H\cap S_1|),\ldots,g'(|H\cap S_k|))\\
&\not\equiv (1,\ldots,1)\enspace, &\pmod{p}
\end{aligned}
\end{equation}
which follows from \eqref{eq:g'mod} and the assumption $(|H\cap S_1|,\ldots,|H\cap S_k|)\not\equiv (0,\ldots,0)\pmod{m}$. Together,~\eqref{eq:setsTo1} and~\eqref{eq:vecsNotTo1} imply property~\ref{item:MKDDiffParity} in Definition~\ref{def:mkdSystem}. Finally, observe that $G'(\mathcal{H})$ is $k$-covering because $\mathcal{H}$ is $k(m-1)$-covering and $G'$ is of level $m-1$. Hence, $G'(\mathcal{H})$ is indeed an $(m,k,k)$-system. 
Together with~\eqref{eq:setsTo1} and~\eqref{eq:vecsNotTo1}, we see that $G'(\mathcal{H})$ fulfills all conditions of Lemma~\ref{lem:GCTFapplication}, and hence we obtain the desired contradiction.
\end{proof}

\section{Barriers for extensions beyond prime powers}
\label{sec:barriers}

In this section, we reveal limits of our techniques by showing that they cannot extend beyond prime power moduli. This points to an interesting structural difference for $m$ being a prime power versus $m$ having at least two different prime factors, and opens up the question whether~\eqref{eq:CCSM} or~\eqref{eq:GCCSM} may be substantially harder for $m$ not being a prime power. This may also shed further light on the complexity of ILPs with a constraint matrix containing subdeterminants that are not prime powers.

When proving correctness of \refEnumD for \eqref{eq:CCSM} and \eqref{eq:GCCSM}, a crucial step that requires the restriction to prime power moduli $m$ is Lemma~\ref{lem:gForPrimePowers}, where we prove that a suitable transformation function $g\colon\mathbb{Z}_{\geq0}\to\mathbb{Z}_{\geq0}$ has the property
\begin{equation}\label{eq:crucialProperty}
g(x) \equiv \begin{cases}
0 \pmod{p} & \text{if } x\equiv 0 \pmod{m},\\
1 \pmod{p} & \text{if } x\not\equiv 0 \pmod{m}
\end{cases}
\end{equation}
for some prime number $p$. The following two theorems present strong implications that result from imposing the above condition with composite moduli $m$.

\begin{theorem}\label{thm:barriersSuperpolyBound}
Let $m,p,r\in\mathbb{Z}_{>0}$, where $p$ is prime, and $m$ is a composite number with $r$ different prime factors. There is a constant $c=c(m)>0$ such that for every cardinality transformation function $g\colon\mathbb{Z}_{\geq 0}\to\mathbb{Z}_{\geq 0}$ fulfilling~\eqref{eq:crucialProperty} with respect to $p$ and $m$, there is a constant $\kappa\in\mathbb{Z}_{\geq 0}$ with the property that for all $n\in\mathbb{Z}$ with $n\geq\kappa$,
\begin{equation}\label{eq:superpolyUpperBoundCTF}
g(n) \geq n^{c\cdot\left(\frac{\log n}{\log\log n}\right)^{r-1}}
\enspace.
\end{equation}
\end{theorem}

Notice that the cardinality transformation functions that we used, which are all of the form described by Corollary~\ref{cor:polyBinomCTF}, have the property that for any constant level, they are polynomially bounded. Hence, the above theorem implies that for an extension beyond prime power moduli based on the cardinality transformation functions we introduced, we would need a superconstant level. As our algorithmic approach relies on Theorems~\ref{thm:targetCTF} and~\ref{thm:EnumDGoodIfNoBadSys}, this would in turn imply that the corresponding enumeration procedure has superconstant depth, prohibiting our algorithm to be efficient.

However, the above theorem does not exclude that there may be other cardinality transformation functions, not covered by Corollary~\ref{cor:polyBinomCTF}, that have constant level and fulfill~\eqref{eq:crucialProperty}. The next theorem rules out this possibility.
More precisely, the next theorem shows that no cardinality transformation function with property~\eqref{eq:crucialProperty} exists even if the level is allowed to depend on $n$, i.e., the size of the ground set, and grows moderately in terms of $n$.
To capture this setting in the following, we allow the level $\ell$ of a cardinality transformation function $g$ to be a function $\ell\colon\mathbb{Z}_{\geq 0} \rightarrow \mathbb{Z}_{\geq 0}$, with the semantics that on any ground set of cardinality $n$, there is a $g$-realizing set transformation of level $\ell(n)$.
To emphasize this difference to our original definition of level, which did not depend on $n$, we will also talk about \emph{generalized level}.

\begin{theorem}\label{thm:barriersLargeLevel}
Let $m,p,r\in\mathbb{Z}_{>0}$ such that $p$ is prime, and $m$ is a composite number with $r$ different prime factors. Every cardinality transformation function $g\colon\mathbb{Z}_{\geq 0}\to\mathbb{Z}_{>0}$ fulfilling~\eqref{eq:crucialProperty} with respect to $p$ and $m$ has a generalized level $\ell$ that satisfies
\begin{equation*}\label{eq:superconstUpperBoundLevel}
\ell = \Omega\left( \left(\frac{\log n}{\log \log n}\right)^{r-1} \right)\enspace.
\end{equation*}
\end{theorem}

We highlight that in the above $\Omega$-notation, $m$ and $p$ are considered to be constant.
The barriers highlighted by the above theorems originate from combinatorial results on set systems with restricted intersections. On the one hand, we have the following classical result by Frankl and Wilson.

\begin{theorem}[Frankl and Wilson~\cite{frankl_1981}]\label{thm:FranklWilson}
Let $p$ be a prime number, let $s\in [p-1]$, and let $\mu_0,\ldots,\mu_s\in\{0,\ldots,p-1\}$ be distinct numbers. Let $\mathcal{H}$ be a set system on a ground set of $n$ elements such that for some $k\in \mathbb{Z}_{\geq 0}$ with $k\equiv \mu_0 \pmod{p}$,
\begin{enumerate}[label=(\roman*),itemsep=-0.2em,topsep=0.2em]
\item $\mathcal{H}$ is a $k$-uniform set system, i.e., $|H|=k$ for all $H\in\mathcal{H}$, and
\item for all distinct $H_1,H_2\in\mathcal{H}$, we have $|H_1\cap H_2| \equiv \mu_i$ for some $i\in[s]$.
\end{enumerate}
Then, $|\mathcal{H}|\leq\binom{n}{s}$.
\end{theorem}

While the above theorem shows that restricting the cardinalities of intersections modulo a prime number reduces the size of the set system to a polynomial in the size of the ground set, the situation changes if the prime modulus is replaced by a composite number. This surprising fact was observed by Grolmusz, who proved the following theorem.

\begin{theorem}[Grolmusz~\cite{grolmusz_2000}]\label{thm:Grolmusz}
Let $m\in\mathbb{Z}_{\geq 0}$ be a composite number with $r>1$ different prime divisors. Then, there is a constant $c_0=c_0(m)>0$ with the property that for every $n\in\mathbb{Z}_{> 0}$, there exists a set system $\mathcal{H}$ on a ground set of $n$ elements such that
\begin{enumerate}[label=(\roman*),itemsep=-0.2em,topsep=0.2em]
\item $|\mathcal{H}|\geq n^{c_0\cdot\left(\frac{\log n}{\log\log n}\right)^{r-1}}$,
\item for all $H\in\mathcal{H}$, we have $|H|\equiv 0\pmod{m}$, and
\item for all distinct $H_1,H_2\in\mathcal{H}$, we have $|H_1\cap H_2| \not\equiv 0\pmod{m}$. 
\end{enumerate}
\end{theorem}

The value of the constant $c_0$ in the above theorem equals roughly $p_r^{-r}$, where $p_r$ is the largest prime divisor of $m$ \cite{grolmusz_2000}, and the constant $c$ in Theorem~\ref{thm:barriersSuperpolyBound} depends on $c_0$. We actually show that $c<c_0$ is a feasible choice. 
The proofs of Theorems~\ref{thm:barriersSuperpolyBound} and~\ref{thm:barriersLargeLevel} follow a common idea. In both, we assume existence of the respective transformation functions, and then use these functions to transform a set system of the type given by Theorem~\ref{thm:Grolmusz} to a new set system. Adjusting the new set systems so that they fulfill the assumptions of Theorem~\ref{thm:FranklWilson} gives an upper bound on their size, and combining these bounds with the lower bound coming from Theorem~\ref{thm:Grolmusz} allows for deducing the results.

\begin{proof}[Proof of Theorem~\ref{thm:barriersSuperpolyBound}]
We show that for every composite number $m$, we can choose any constant $c<c_0$, where $c_0$ is the corresponding constant guaranteed by Theorem~\ref{thm:Grolmusz}.
Let $m$ be a composite number, and let $g\colon\mathbb{Z}_{\geq0}\to\mathbb{Z}_{\geq0}$ be a cardinality transformation function fulfilling property~\eqref{eq:crucialProperty} with respect to the prime number $p$ and $m$.
For $n\in\mathbb{Z}$, let $\mathcal{H}$ be a set system on a ground set of size $n$ fulfilling the properties listed in Theorem~\ref{thm:Grolmusz} with respect to the composite number $m$. The set system $\mathcal{H}$ is not necessarily a uniform set system, but it contains a large uniform subsystem. To see this, define $\mathcal{H}_i=\{H\in\mathcal{H} \mid |H|=i\}$ for $i\in[n]$ and let $\ell\in\argmax_{i\in[n]} |\mathcal{H}_i|$. Then $\mathcal{H}_\ell$ is an $\ell$-uniform set system with $|\mathcal{H}_\ell|\geqslant\frac1n|\mathcal{H}|$.

Let $G$ be a $g$-realizing set system transformation function and consider the set system $G(\mathcal{H}_\ell)$. We claim that $G(\mathcal{H}_\ell)$ is a set system on a ground set of size $g(n)$ that fulfills the assumptions of Theorem~\ref{thm:FranklWilson}. 
Obviously, $G(\mathcal{H}_\ell)$ is a uniform system, as for all $H\in\mathcal{H}_\ell$, we have $|H|=\ell$ and hence $|G(H)|=g(|H|)=g(\ell)$, so the set system is $g(\ell)$-uniform. Moreover, as by assumption, $|H|\equiv0\pmod{m}$, property~\eqref{eq:crucialProperty} implies $g(\ell)=g(|H|)\equiv 0\pmod{p}$. Note that for any two distinct sets $H_1,H_2\in\mathcal{H}$, we have
\begin{equation}\label{eq:intersectionInImage}
|G(H_1)\cap G(H_2)| = |G(H_1\cap H_2)| = g(|H_1\cap H_2|) \equiv 1 \pmod{p} \enspace,
\end{equation}
where we used the assumption that $|H_1\cap H_2|\not\equiv 0 \pmod{m}$ for all distinct $H_1,H_2\in\mathcal{H}$, and property~\eqref{eq:crucialProperty}. Hence, $G(\mathcal{H}_\ell)$ fulfills the conditions of Theorem~\ref{thm:FranklWilson} with $s=1$, $\mu_0=0$, $\mu_1=1$, and $k=g(\ell)$. As $\mathcal{H}_\ell$ is a system on a ground set of size $n$, $G(\mathcal{H}_\ell)$ is one on a ground set of size $g(n)$. By Theorem~\ref{thm:FranklWilson}, we thus obtain the upper bound $|G(\mathcal{H}_\ell)| \leq g(n)$.

Note that if in \eqref{eq:intersectionInImage}, $G(H_1)$ and $G(H_2)$ are not distinct, then $|G(H_1)\cap G(H_2)|=|G(H_1)|=g(|H_1|)\equiv 0\pmod{p}$. This contradicts \eqref{eq:intersectionInImage}, hence $H_1=H_2$ whenever $G(H_1)=G(H_2)$, so $G$ is injective when restricting its domain to $\mathcal{H}$. Injectivity of $G$ on $\mathcal{H}$ implies
\begin{equation*}
|G(\mathcal{H}_\ell)|=|\mathcal{H}_\ell| \geq \frac{|\mathcal{H}|}{n} \geq n^{c_0\cdot\left(\frac{\log n}{\log\log n}\right)^{r-1}-1}\enspace.
\end{equation*}
Combining the obtained upper and lower bounds on $|G(\mathcal{H}_\ell)|$, we get the inequality
\begin{equation*}
g(n) \geq n^{c_0\cdot\left(\frac{\log n}{\log\log n}\right)^{r-1}-1}\enspace.
\end{equation*}
From the above, it is easy to see that whenever $c<c_0$, there exists a constant $\kappa\in\mathbb{Z}_{\geq0}$ such that every $n\in\mathbb{Z}$ with $n\geq\kappa$ satisfies
\begin{equation*}
g(n) \geq n^{c\cdot\left(\frac{\log n}{\log\log n}\right)^{r-1}}\enspace.\qedhere
\end{equation*}
\end{proof}

To present the proof of Theorem~\ref{thm:barriersLargeLevel}, we introduce the concept of \emph{atoms} of a set system. When applying a set system transformation of level $\ell$, we cannot directly bound the size of the ground set of the new set system. Nonetheless, we can show that the size of the new ground set can be reduced to a polynomial in the size of the ground set of the initial set system without loosing the system's structure. The key ingredient for this procedure is bounding the number of atoms in the transformed set system. This is formalized in Lemma~\ref{lem:atoms} and will be an important building block for the proof of Theorem~\ref{thm:barriersLargeLevel}.

\begin{definition}
Let $\mathcal{H}$ be a set system on a finite ground set $N$. A non-empty set $A\subseteq N$ is an \emph{atom} of $\mathcal{H}$ if it is a maximal set with the property that for all $H\in\mathcal{H}$, we have $A\subseteq H$ or $A\subseteq N\setminus H$.
\end{definition}

In particular, the above definition implies that two elements of the ground set $N$ are not in the same atom if and only if the set system $\mathcal{H}$ contains a set separating the two elements.

\begin{lemma}\label{lem:atoms}
Let $g\colon\mathbb{Z}_{\geq0}\to\mathbb{Z}_{\geq0}$ be a cardinality transformation function of level $\ell\in\mathbb{Z}_{\geq 0}$. Let $N$ be a set of size $n$, and let $G$ be a $g$-realizing set transformation function for the ground set $N$. Then, $G(2^N)$ has at most $1+\ell n^\ell$ many atoms.
\end{lemma}

\begin{proof}
Since $g$ is of level $\ell$, for every $w\in W$ there is a set $S_w \subseteq N$ with $w\in G(S_w)$ and $|S_w|\leq \ell$. Among all such sets, let $S_w$ be one that is inclusion-wise minimal. (Actually, one can observe that $S_w$ is unique; however, we do not need this later.)
Moreover, we denote by $A_w\subseteq G(N)$ the atom of $G(2^N)$ containing $w$.

Because $|S_w|\leq \ell$ for all $w\in W$, the number of different sets $S_w$ can be bounded from above by the number of subsets of $N$ of size at most $\ell$, i.e.,
\begin{equation*}
|\{ S_w \mid w\in W \}| \leq \sum_{i=0}^\ell \binom{n}{i} \leq 1+\ell n^\ell\enspace. 
\end{equation*}
To finish the proof, we show that the map $A_w\mapsto S_w$ is an injection. If so, we get $|\{ A_w \mid w\in W \}| \leq |\{ S_w \mid w\in W \}|$, which, together with the above bound, proves the lemma. To see injectivity, let $w_1,w_2\in W$ with $A_{w_1}\neq A_{w_2}$, i.e., $w_1$ and $w_2$ are not in the same atom. Then, there is a set $G(S)\in G(2^N)$ separating the two elements, for some set $S\subseteq N$. Without loss of generality, assume that $w_1\in G(S)$, while $w_2\not\in G(S)$.

On the one hand, this implies $w_1\in G(S)\cap G(S_{w_1})=G(S\cap S_{w_1})$, hence by minimality of $S_{w_1}$, we get $S\cap S_{w_1}=S_{w_1}$. On the other hand, we have $w_2\notin G(S)\cap G(S_{w_2})=G(S\cap S_{w_2})$, hence $S\cap S_{w_2}\subsetneq S_{w_2}$. This implies $S_{w_1}\neq S_{w_2}$, and hence injectivity of the map $A_w\mapsto S_w$, as desired.
\end{proof}

Before we start the proof of Theorem~\ref{thm:barriersLargeLevel}, we remark that both the statement and the proof of the previous lemma remain unchanged even if we allow for using the notion of generalized level, thus leading to an upper bound of $1+\ell(n) n^{\ell(n)}$ many atoms.

\begin{proof}[Proof of Theorem~\ref{thm:barriersLargeLevel}]
Let $m$ be a composite number, and let $g\colon\mathbb{Z}_{\geq0}\to\mathbb{Z}_{\geq0}$ be a cardinality transformation function that fulfills the property~\eqref{eq:crucialProperty} for some prime number $p$. Moreover, let $\ell\colon\mathbb{Z}\to\mathbb{Z}$ denote the (generalized) level of $g$.

We know that for every $n\in\mathbb{Z}_{\geq 0}$, there exists a set system $\mathcal{H}$ on a ground set $N$ of size $n$ fulfilling the properties listed in Theorem~\ref{thm:Grolmusz} with respect to the composite number $m$. Let $G$ be a $g$-realizing set transformation function on $N$ and consider $G(\mathcal{H})$. Note that by property~\eqref{eq:crucialProperty} and the assumptions on $\mathcal{H}$, every set $G(H) \in G(\mathcal{H})$ satisfies $|G(H)|\equiv 0\pmod{p}$, while for every two distinct sets $G(H_1),G(H_2)\in G(\mathcal{H})$, we have $|G(H_1)\cap G(H_2)|=|G(H_1\cap H_2)|=g(|H_1\cap H_2|)\equiv 1\pmod{p}$.

As we are only interested in the size of sets in $G(\mathcal{H})$ and their intersections $\bmod\ p$, and because every such set is a disjoint union of atoms of $G(\mathcal{H})$, we can delete elements of the ground set in the following way without loosing the observed properties. For every atom $A$ of $G(\mathcal{H})$, if $|A|\equiv a\pmod{p}$ with $a\in\{1,\ldots,p\}$, we can delete any $|A|-a$ elements of $A$ from the ground set $G(N)$, and update the sets in $G(\mathcal{H})$ correspondingly by removing the deleted elements from all sets containing them. By doing so, we thus obtain a new set system $\mathcal{I}$ with atoms of cardinality at most $p$. Note that none of the atoms were deleted completely, and thus distinct sets in $G(\mathcal{H})$ before the deletion of elements remain distinct after the deletion, i.e., in $\mathcal{I}$. Thus
\begin{equation}\label{eq:sameCardAfterDel}
|\mathcal{I}| = |G(\mathcal{H})|\enspace.
\end{equation}
In particular, the number of atoms in $\mathcal{I}$ equals the number of atoms in $G(\mathcal{H})$, which, by Lemma~\ref{lem:atoms}, is bounded by $1+\ell n^\ell$. Altogether, $\mathcal{I}$ is a set system on a ground set of size at most $p(1+\ell n^\ell)$ such that $|I|\equiv 0\pmod{p}$ for all $I\in\mathcal{I}$, and $|I_1\cap I_2|\equiv 1\pmod{p}$ for all distinct sets $I_1, I_2\in \mathcal{I}$.

In order to apply Theorem~\ref{thm:FranklWilson}, we need a large uniform subsystem of $\mathcal{I}$. Thereto, let $\mathcal{I}_i=\{I\in\mathcal{I} \mid |I|=i \}$ for $i\in[p(1+\ell n^\ell)]$ be all uniform subsystems, and let $k\in[p(1+\ell n^\ell)]$ be such that $|\mathcal{I}_k|$ is the one of maximum cardinality.
The $k$-uniform set system $\mathcal{I}_k$ satisfies the assumptions of Theorem~\ref{thm:FranklWilson} with $s=1$, $\mu_0=0$ and $\mu_1=1$, so by the same theorem, we get
\begin{equation}\label{eq:IkSmall}
|\mathcal{I}_k|\leq p(1+\ell n^\ell)\enspace.
\end{equation}

For a lower bound, first note that that $|\mathcal{I}_{k}| \geq \frac{1}{p(1+\ell n^\ell)}\,|\mathcal{I}|$.
Furthermore, as already observed in the proof of Theorem~\ref{thm:barriersSuperpolyBound}, $G$ is injective over the domain $\mathcal{H}$, and thus, we have $|G(\mathcal{H})|=|\mathcal{H}|$.
Putting this together and using the lower bound on the size of $\mathcal{H}$, we get
\begin{equation*}
|\mathcal{I}_{k}| \geq \frac{|\mathcal{I}|}{p(1+\ell n^\ell)} = \frac{|\mathcal{H}|}{p(1+\ell n^\ell)} \geq \frac{n^{c_0\cdot\left(\frac{\log n}{\log\log n}\right)^{r-1}}}{p(1+\ell n^\ell)} \enspace,
\end{equation*}
where the equality follows from~\eqref{eq:sameCardAfterDel} and $|G(\mathcal{H})|=|\mathcal{H}|$.
Combining this with~\eqref{eq:IkSmall} and rearranging terms, we obtain
\begin{equation*}
1+\ell n^\ell \geq \frac{n^{\frac{c_0}{2}\left(\frac{\log n}{\log\log n}\right)^{r-1}}}{p}\enspace.
\end{equation*}
Note that when transforming a system on a ground set of size $n$, the level is always at most $n$, i.e., $n\geq \ell$. Using this and absorbing constants into the asymptotic notation, we obtain
\begin{equation*}
n^{\ell+1} \geq \ell n^\ell = n^{\Omega\left(\left(\frac{\log n}{\log\log n}\right)^{r-1}\right)}\enspace,
\end{equation*}
which implies the desired $\ell=\Omega\left(\left(\frac{\log n}{\log\log n}\right)^{r-1}\right)$.
\end{proof}

\section{Minimality of the enumeration depth $d$}\label{sec:existenceMMm2Systems}

In this section, we show that for any $m \in \mathbb{Z}_{>0}$, \refEnumD[d] does in general not solve~\eqref{eq:CCSM} with modulus $m$ correctly if $d < m-1$. This shows in particular that our choice $d=m-1$ of the depth of \refEnumD[d] is the smallest depth for which \refEnumD[d] successfully solves~\eqref{eq:CCSM} for prime power moduli $m$.
This also implies the existence of $(m,m-2)$-systems for $m\in \mathbb{Z}_{>0}$; this follows from Theorem~\ref{thm:EnumDGoodIfNoBadSys}, but can also be seen directly from our construction.

We show that $d\geq m-1$ is necessary by constructing an explicit example where $d=m-2$ is not enough for \refEnumD to solve \eqref{eq:CCSM}. Thereto, let $N=\{0,1,\ldots,n\}$ for some $n\in\mathbb{Z}$ with $n\geq m$. Consider the lattice $\mathcal{L}=2^N$ and define the modular (and thus also submodular) function $f\colon\mathcal{L}\to\mathbb{Z}$ by
\begin{equation*}
f(S) = \begin{cases}
|S| & \text{if } 0\notin S,\\
|S|-1-m & \text{if } 0\in S
\end{cases}
\end{equation*}
for all $S\subseteq N$. This function is indeed modular since it assigns weight $-m$ to the element $0$, and weight $1$ to all other elements of $N$, and the weight $f(S)$ of a subset $S\subseteq N$ is obtained by summing the weights of all elements. The \eqref{eq:CCSM} problem that we consider is minimizing the function $f$ over the subfamily
\begin{equation*}
\mathcal{F} = \{ S\in\mathcal{L} \mid |S|\equiv 0 \pmod{m} \}\enspace.
\end{equation*}
It is easy to see that $\min\{f(S)\mid S\in\mathcal{F}\}=-1$, with minimizers being precisely all $m$-element subsets of $N$ containing $0$.

However, \refEnumD[m-2] does not solve this \eqref{eq:CCSM} problem, as we now show. Consider a step of \refEnumD[m-2], i.e., fix $A,B\subseteq N$ with $|A|,|B|\leq m-2$ and $A\cap B=\emptyset$. It is easy to see that
\begin{equation*}
\argmin \{f(S) \mid S\in\mathcal{L}_{AB} \} = \begin{cases}
\{A\} & \text{if } 0\in B,\\
\{A\cup\{0\}\} & \text{if } 0\notin B.
\end{cases}
\end{equation*}
In all cases, the minimizers found will be sets of size at most $m-1$, while the actual minimizers of $f$ over $\mathcal{F}$ are of size $m$. So \refEnumD[m-2] does indeed not solve this \eqref{eq:CCSM} problem.

As indicated above, existence of an $(m,m-2)$-system thus follows from Theorem~\ref{thm:EnumDGoodIfNoBadSys}. This system can be constructed by following the proof of Lemma~\ref{lem:notDGoodToSys}, resulting in a system on a ground set of $m$ elements containing all sets of size at most $m-1$ containing a fixed element.

\section{Conclusions}

We presented a new approach to deal with submodular function minimization problems under congruency constraints. The core of our approach is the analysis of a very natural algorithm, that enumerates over small subsets of elements to be included, respectively excluded, in a minimizer.
Our analysis reduces the correctness of this procedure to a purely combinatorial question about the nonexistence of certain set systems, which we can settle when the modulus of the involved congruency constraints is a prime power, by using techniques from Combinatorics and Number Theory. This leads to polynomial time algorithms for~\eqref{eq:CCSM} and \eqref{eq:GCCSM} when the modulus $m$ is a prime power bounded by a constant.
The techniques we introduced to disprove such set systems can be seen as a general framework, which we hope may be useful for future extensions to solve submodular function minimization problems under even more general constraint families.

It remains open whether \eqref{eq:CCSM} and \eqref{eq:GCCSM} can be solved efficiently for a constant modulus $m$ that is not a prime power. However, as we highlighted in Section~\ref{sec:barriers}, this would require new ingredients. A recent construction by Gopi \cite{gopi2017systems}, which was found after submission of this work, strengthens the barriers pointed out in Section~\ref{sec:barriers} by showing that $(m,m-1)$-systems do actually exist if $m$ is not a prime power. Gopi's construction is based on results by Barrington, Beigel, and Rudich~\cite{barrington_1994_representing} on the representation of boolean functions. Results in~\cite{barrington_1994_representing} were also leveraged by Grolmusz in his proof of Theorem~\ref{thm:Grolmusz}.

Moreover, we highlight that our proofs imply that our enumeration algorithm, when applied to~\eqref{eq:CCSM} or~\eqref{eq:GCCSM}, enumerates \emph{all} minimal optimal solutions. In particular, this shows that in the discussed settings where our approach finds an optimal minimal solution in polynomial time, the total number of minimal solutions is polynomially bounded.

Since both \eqref{eq:CCSM} for $m\geq 3$ and \eqref{eq:GCCSM} for $m\geq 2$ are not captured by triple or parity families, and neither do they generalize these families, it remains open to find a common generalization. In particular, submodular function minimization over the intersection of a constant number of parity families would be such a common generalization. It remains open whether this problem can be solved efficiently.

\section*{Acknowledgments}
We thank Karthekeyan Chandrasekaran for interesting discussions on related topics. Moreover, we are grateful to the anonymous referees for various comments and suggestions that helped to improve the quality of the presentation.

\bibliographystyle{plain}

\end{document}